\documentclass[11pt]{article}
\usepackage{amssymb,amsmath,epsfig,mst-stylefile,harvard}

\newcommand{\bbZ}{{\Bbb Z}}
\newcommand{\bbR}{{\Bbb R}}
\newcommand{\bbN}{{\Bbb N}}

\renewcommand{\cite}{\citeyear}

\begin{document}

\title{On the wavelet-based simulation of anomalous diffusion
\thanks{The first author was supported in part by the Louisiana Board of Regents award LEQSF(2008-11)-RD-A-23.
The second author was supported in part by the NSF grant DMS-0714939.}
\thanks{The authors would like to thank Vladas Pipiras for his comments on this work.}
\thanks{{\em AMS Subject classification}. Primary:  60G10, 60H35, 42C40, 60G15.}
\thanks{{\em Keywords and phrases}: simulation; wavelets; anomalous diffusion; generalized Langevin equation; fractional Ornstein-Uhlenbeck process.} }

\author{Gustavo Didier \\ Tulane University  \and John Fricks\\ Penn State University}

\bibliographystyle{agsm}

\maketitle

\begin{abstract}

The field of microrheology is based on experiments involving particle diffusion. Microscopic tracer beads are placed into a non-Newtonian fluid and tracked using high speed video capture and light microscopy. The modeling of the behavior of these beads is now an active scientific area which demands multiple stochastic and statistical methods.

We propose an approximate wavelet-based simulation technique for two classes of continuous time anomalous diffusion models, the fractional Ornstein-Uhlenbeck process and the fractional generalized Langevin equation. The proposed algorithm is an iterative method that provides approximate discretizations that converge quickly and in an appropriate sense to the continuous time target process. As compared to previous works, it covers cases where the natural discretization of the target process does not have closed form in the time domain. Moreover, we propose to minimize the border effect via smoothing.

\end{abstract}



\section{Introduction}\label{s:intro}

The characterization of particle diffusion is a classical problem in physics and probability theory dating back to the early work of Einstein on Brownian motion. The case of Newtonian fluids, which includes water, is now well-understood; however, the case of complex fluids is an active area of experimental and theoretical inquiry. Of particular interest are fluids of biological origin such as mucus, which is both highly viscous due to the presence of compounds such as mucin, a high molecular weight protein that enters into the composition of human mucus, and elastic due to the cross-linking of the mucin. Models of diffusion in biological fluids have been developed for many pharmaceutical and medical applications (e.g., Lai et al.\ \cite{lai:wang:hanes:2009}, Suk and Lai \ \cite{suk:lai:2009}, Suh et al.\ \cite{suh:dawson:hanes:2005}, Dixit et al.\ \cite{dixit:ross:goldman:holzbaur:2008}). The related field of microrheology is based on experiments in which tracer beads, whose radii range from tens of nanometers to micrometers, are placed into the complex fluid and tracked at millisecond sampling rates using high speed video capture and light microscopy (e.g., Mason and Weitz \cite{mason:weitz:1995}). Tens of those beads can be tracked at a time, and the data is preprocessed using particle tracking software to yield the position $X(t)$, as a function of time, of the individual diffusing particles. The study of the behavior of these beads generates a wide range of interconnected statistical problems such as physically-informed probabilistic modeling, time and spectral domain inference, the characterization of first passage times, and simulation. For a detailed description of these problems, see Didier et al.\ \cite{didier:mckinley:hill:fricks:2012} and references therein.

In the physics literature, a position process is called diffusive if its second moment satisfies Einstein's Brownian diffusion law, i.e.,
\begin{equation}\label{e:diffusivity_long_term}
EX^2(t) \sim t^{\alpha}, \quad t \rightarrow \infty,
\end{equation}
with $\alpha = 1$. However, if $\alpha < 1$ or $>1$, then it is said to be anomalously diffusive (sub- and superdiffusive, respectively; see Kou \cite{kou2008sat}, and references therein). In this paper, our goal is to provide a fast, approximate wavelet-based method for two classes of continuous time anomalous diffusion models. We approach this problem with wavelet methods along the lines of the papers of Didier and Pipiras \cite{didier:pipiras:2008,didier:pipiras:2010}, as well as many other works including Meyer et al.\ \cite{MST:1999}, Zhang and Walter \cite{zhang:walter:1994}, Sellan \cite{sellan:1995}, Pipiras \cite{pipiras:2004}. In particular, the work in Didier and Pipiras \cite{didier:pipiras:2008} explores the idea that it is natural to consider simulating a continuous time process by means of its discretization at a given scale (related to a grid) and to use a fast wavelet algorithm to connect discretizations over multiple scales. This approach works well with processes such as the Ornstein-Uhlenbeck (OU). Due to its Markovian nature, its (exact) discretization over any regular grid is an AR(1) sequence, which displays a rather simple correlation structure, available in closed form. By contrast, the discretization of anomalous diffusion models usually exhibit intricate correlation structures (see Pipiras \cite{pipiras:2005}, section 4, on a similar problem with fBm and its simulation). In this paper, we propose a simple method to generate (approximate) discretizations of certain classes of continuous time anomalous diffusion models whose spectral densities exhibit fractional behavior at the origin and for which correlation structures are not available in closed form. These sequences of approximate discretizations quickly converge to the continuous time process. Furthermore, we propose a way to smooth the associated simulation filters in the Fourier domain so to accelerate their time domain decay and thus minimize the effect of the truncation of infinite length filters upon computational implementation. While the smoothing is not necessary to obtain convergence, it improves the stochastic accuracy of the method in practice. We now explain these issues in more detail.

For a viscous fluid such as water, the velocity $V$ of a free particle satisfies the Langevin equation. The stationary solution for $V$ gives the well-known OU process. In order to model the non-Markovian nature of non-Newtonian, viscoelastic fluids such as many biological fluids, one generalizes the Langevin equation by including a memory kernel $\Gamma(t)$. The resulting generalized Langevin equation (GLE) for the velocity process of a free particle is
\begin{equation} \label{eq:gle}
	m\frac{d}{dt} V(t) = -\zeta \int_{-\infty}^t \Gamma(t-s) V(s) ds + F(t), \quad t \geq 0,
\end{equation}
where $m$ is the particle mass, $\zeta$ is the friction constant, $k_B$ is the Boltzmann constant. The term $F(t)$ is a stationary Gaussian process with autocorrelation function $E{F(t)F(s)} = k_B \tau \zeta \Gamma(|t-s|)$, where $\tau$ is the temperature. The special case of the classical Langevin equation is obtained by setting $\Gamma$ to a Dirac delta distribution. In this paper, we are particularly interested in the (subdiffusive) fractional GLE (fGLE), since its correlation structure is now well-studied (see Kou \cite{kou2008sat}). This corresponds to setting $F(t)dt$ to $dB_{H}(t)$ up to a constant, where $B_{H}(t)$ is a fractional Brownian motion (fBm; see, for instance, Taqqu \cite{taqqu:2003}). 
The fractional Ornstein-Uhlenbeck process (fOU) is another model for anomalous diffusion of interest in this paper. It is the a.s.\ continuous solution to the fBm-driven Langevin equation
\begin{equation} \label{e:fOU_SDE}
dV(t) = - \zeta V(t) dt + \sigma dB_{H}(t), \quad t \geq 0, \quad 0 < H < 1
\end{equation}
(see Cheridito et al.\ \cite{cheridito:kawaguchi:maejima:2003}, Prakasa Rao \cite{prakasarao:2010}).  Its simulation is interesting in its own right, since the fOU is a model for both sub- and superdiffusion. In addition, simulation of the fOU can be viewed as a step towards the simulation of the full fGLE, both for analytical convenience and due to their similar correlation structures.  For a comparison between these processes, see Section \ref{s:Fourier_integ}.

Our proposed simulation procedure is based on a wavelet analysis of the velocity process $V$. This analysis encompasses three components:
\begin{enumerate}
\item [$(i)$] a wavelet-based decomposition of $V(t)$;
\item [$(ii)$] a sequence of stationary discrete time processes $V_j =
\{V_{j,n}\}_{n \in \mathbb{Z}}$,
\begin{equation}\label{e:V_j_Wold}
V_{j,k} = \sum^{\infty}_{n=-\infty} g_{j,n} \xi_{k-n},  \quad \{\xi_{k}\} \stackrel{\textnormal{i.i.d.}}\sim \textnormal{WN}(0,1),
\end{equation}
that can be thought of as (exact or approximate) discretizations of the continuous time process $V$ at the (wavelet) scale
$2^{-j}$;
\item [$(iii)$] a Fast Wavelet Transform (FWT)-like algorithm relating
$V_{j}$ across different scales.
\end{enumerate}
The theoretical backbone of the simulation technique, which makes use of a Meyer multiresolution analysis (MRA; see Mallat \cite{mallat:1999}), is given by $(i)$. More details can be found in Didier and Pipiras \cite{didier:pipiras:2008}. 
The simulation procedure itself, which is the subject of this paper, can be directly expressed as components $(ii)$ and $(iii)$ (see Section \ref{s:choice} for a schematic description of the algorithm). The latter is based on the Cram\'{e}r-Wold Fourier domain representation of the stationary velocity process
\begin{equation}\label{e:Cramer-Wold_spec}
V(t) = \int_{\bbR}e^{itx}\widehat{g}(x) \widetilde{B}(dx),
\end{equation}
where $\widehat{g}(x) \in L^2(\bbR)$ is called a spectral filter and $\widetilde{B}(dx)$ is a complex valued Brownian measure (see \eqref{e:Brownian_measure}). 
The simulation method relies on designing an appropriate sequence of discrete time filters $g_{j}$, $j=0,1,...,J$, where $J$ is the finest scale chosen (component $(ii)$ above). Then, the filters $g_{j}$ are used to generate an induced sequence via a wavelet filter-based recursion of discretizations $V_j$, at scale $j \in \bbN$ (component $(iii)$ above). We can show that
\begin{equation}\label{e:wavelet crime}
2^{J/2}V_{J,\lfloor 2^{J}t \rfloor } \approx V(t), \quad J \rightarrow \infty,
\end{equation}
in a suitable sense, where $\lfloor x \rfloor$ denotes the integer part of $x \in \bbR$. This in turn leads to a natural approximation of the position process $X(t)$ as a Riemann sum (see Theorem \ref{t:VJ_conv_unif} and Corollary \ref{c:X_Riemann} for the rigorous statements).

The choice of the sequence of discretization filters $g_j$ is the key requirement. Heuristically, it should be such that
\begin{equation}\label{e:GJ_approx_1}
G_{j}(2^{-j}x) = \frac{\widehat{g}_{j}(2^{-j}x)}{\widehat{g}(x)}\approx 1, \quad j \rightarrow \infty, \quad x \in \bbR,
\end{equation}
where $\widehat{g}_{j}(x)$ is periodically extended to $\bbR$ (see Assumption 3 in Section \ref{s:choice} for the actual meaning of the expression \eqref{e:GJ_approx_1}). Intuitively, the discretization filter approximates, up to a scaling factor, the continuous time filter as the scale becomes finer and finer. In this regard, the case of the OU process is quite special. In the simulation framework of this paper, the specification
\begin{equation}\label{e:GJ_OU_exact}
G_{j}(x) = 2^{-j} \frac{(1 - e^{-\zeta 2^{-j}}e^{- i x})^{-1}}{(\zeta + 2^j i x)^{-1}}
\end{equation}
is then a natural choice, since $(1 - e^{-\zeta 2^{-j}}e^{- i x})^{-1}$ is a spectral filter of the associated discrete time AR(1) sequence at scale $j$. Note that the heuristic relation (\ref{e:GJ_approx_1}) is indeed satisfied in a pointwise sense. 
However as a rule, the discretization of a continuous time process leads to a substantially more intricate expression for the spectral density, such as for fBm, the fGLE and the fOU and unlike the case of the OU process. Though expressions for the covariance structure are available for the fGLE and the fOU, they do not appear in closed form and will require numerical methods (see Corollary \ref{c:increm_subdiff_specdens}).

The main contribution of this paper is to go beyond Didier and Pipiras \cite{didier:pipiras:2008} in the sense of proposing a simple method to obtain discretization spectral filters (i.e., $G_j$) for the velocity process of the fGLE and the fOU. We exchange exact discretization for mathematical and computational manageability without hindering convergence. The discretization sequences $\{V_{j,n}\}$  obtained are approximate discretizations (at scale $j$) of the continuous time process $V$, but we show that the convergence on finite intervals (property \eqref{e:wavelet crime}) still holds. The method truncates the continuous time spectral filter in the Fourier domain and generates (non-causal) wavelet simulation filters. One should note that the fast time domain decay of filters is a quite desirable property because upon computational implementation, infinite length filters used in convolution-based algorithms must be truncated (the border effect). Therefore, time domain filters with ``lighter tails" tend to improve the stochastic accuracy of the simulation method in practice. So, we show that the proposed wavelet filters also display the property of quadratic decay obtained otherwise by exact discretization. Moreover, we also propose a way to smooth such filters in the Fourier domain to further accelerate their time domain decay in computational practice while preserving the convergence property \eqref{e:wavelet crime}.

Moreover, the simulation procedure shares the positive properties of Didier and Pipiras \cite{didier:pipiras:2008}, such as: it is computationally fast, potentially reaching complexity $O(N)$, since it is based on a Fast Wavelet Transform-like algorithm; it provides iterative discretizations that converge uniformly over compact intervals a.s.; the convergence speed is exponentially fast and depends on the sample path smoothness of the limiting process. It is also iterative both intensively and extensively. In other words, a generated discretization at scale $J$ over a compact interval $[0,T]$, $T \in \bbR_+$, can be used to generate a finer discretization at scale $J+1$ over $[0,T]$ or some expanded interval $[0,T + \chi]$, $\chi \in \bbR_+$. Our simulation procedure also serves to simulate the position process $X$ with the same convergence rate as the velocity process $V$ over compact intervals. Moreover, the method is not intrinsically Gaussian, although Gaussian processes are the primary focus of this paper.

The paper is structured as follows. In Section \ref{s:filters_via_trunc}, we introduce the Fourier domain integral representations for the velocity processes $V$ for the fGLE and the fOU process, describe the simulation method, and develop the discretization filters that enter into the simulation procedure. In Section \ref{s:eval_sim}, we evaluate the accuracy of the wavelet-based simulation method in comparison to other, exact methods. Appendices \ref{s:adaptation_proofs} and \ref{s:aux} contain all the proofs and auxiliary results, respectively. Appendix \ref{s:tables} contains the tables with the simulation results, while Appendix \ref{s:accuracy_numerical_integration} shows a study of the numerical accuracy of the computational techniques. For clarity, pseudocode for the simulation method is provided in Appendix \ref{s:pseudocode}.

\begin{remark}
In wavelet terminology, the sequences $\{V_{J,k}\}$ would be called approximation coefficients, and in this sense they would make up approximations to the continuous time process $V(t)$. However, to avoid confusion with the predominantly deterministic and computational uses of the word ``approximation" in this paper, we opted for only calling $\{V_{J,k}\}$ ``discretizations".
\end{remark}

%

%

\section{Filters in the spectral domain}\label{s:filters_via_trunc}

In this section, we propose approximate discretization filters with the purpose of simulation. The focus is on the fGLE and fOU, but the OU process will be revisited frequently in order to contrast exact and approximate discretization procedures.

\subsection{Fourier domain representations}\label{s:Fourier_integ}

First, we express the integral representations which will be the basis for the construction of the simulation procedures for each class of processes. It is convenient to rewrite the velocity process $V$ as a Fourier domain stochastic integral with respect to a Brownian measure.  So, we define $\widetilde{B}_1$, $\widetilde{B}_2$  as two real-valued Brownian motions and $\widetilde{B}(dx) = \widetilde{B}_1(dx) + i \widetilde{B}_2(dx)$ as the induced random measure satisfying
\begin{equation}\label{e:Brownian_measure}
\widetilde{B}(-dx) = -\overline{\widetilde{B}(dx)} \hspace{2mm}\textnormal{a.s.}, \quad E|\widetilde{B}(dx)|^2 = dx.
\end{equation}

\begin{remark}
Throughout the paper, the Fourier transform of either a discrete or continuous time function/filter $g$ is denoted by $\widehat{g}$. We use the parameter $\delta$ to represent the fractional behavior of a spectral density around the origin. In other words,
\begin{equation}\label{e:specdens_powerlaw}
|\widehat{g}(x)|^2 \sim x^{-2 \delta}, \quad -\frac{1}{2}< \delta < \frac{1}{2}.
\end{equation}
If $\delta > 0$ or $\delta < 0$ , then the process is said to be long range dependent or antipersistent, respectively. As shown in this section, both the fOU and the fGLE have spectral densities that satisfy \eqref{e:specdens_powerlaw}. Moreover, in both cases we will define
\begin{equation}\label{e:d=H-1/2}
d = H - 1/2,
\end{equation}
where $H$ is the Hurst parameter of the driving fBm. Note, however, that for the fOU, $\delta = d$, whereas for the fGLE, $\delta = -d$. In both cases, $\alpha$ in \eqref{e:diffusivity_long_term} and $\delta$ are connected by means of the relation $\alpha = 1 + 2 \delta$ (see Didier et al.\ \cite{didier:mckinley:hill:fricks:2012}).
\end{remark}

The fOU process admits the spectral representation \eqref{e:Cramer-Wold_spec} with
\begin{equation}\label{e:fracOU_spec}
\widehat{g}(x) = \sigma \sqrt{\Gamma(2d+2) \sin(\pi (d+1/2))} \frac{1}{\sqrt{ \zeta^2 + x^2 }} |x|^{-d}, \quad -\frac{1}{2} < d < \frac{1}{2},
\end{equation}
from Proposition \ref{p:V_gle_spec_repres} of Cheridito et al.\ \cite{cheridito:kawaguchi:maejima:2003}, pp.\ 5-8. The integral representation of the OU process can be obtained by setting $d=0$ in \eqref{e:fracOU_spec}, where \eqref{e:d=H-1/2} holds and $H$ is as in \eqref{e:fOU_SDE}.  The velocity process for the fGLE \eqref{eq:gle} can be represented as in \eqref{e:Cramer-Wold_spec} with
\begin{equation}\label{p:V_gle_spec_repres}
\widehat{g}(x) = c(d) \frac{1}{(\gamma_0 + \gamma_1 |x|^{\beta} + \gamma_2 |x|^{2\beta})^{1/2}} |x|^{d}, \quad 0 < d < \frac{1}{2},
\end{equation}
where $\beta = 1 + 2 d$, and the constants are defined by
$$
\gamma_0 = a^2 + b^2, \quad \gamma_1 = 2bm, \quad \gamma = m^2, \quad c(d) = \sqrt{2 \zeta k_B T \hspace{1mm}\Gamma(2d+2)\sin(\pi(d+1/2))}
$$
and $a = \zeta \Gamma(2d+2)\sin(\pi(d+1/2))$, $b = \zeta \Gamma(2d+2)\cos(\pi(d+1/2))$. Note that $\zeta > 0$, and $0 < d < 1/2$, so $a > 0$ and $b < 0$.

The expression \eqref{p:V_gle_spec_repres} can be established by following the techniques of Kou \cite{kou2008sat}, Theorem 2.1. One begins by expressing the velocity process of a fBm-driven free particle \eqref{eq:gle} in terms of a pathwise defined Riemann-Stieltjes integral
\begin{equation}\label{e:V_gle_def}
V(t) = \sqrt{2 \zeta k_B \tau}\int_{\bbR}r(t-u)B_{H}(du).
\end{equation}
The time domain filter $r$ is given through the inverse Fourier transform
$$
r(t) = \frac{1}{2\pi}\int_{\bbR}\frac{1}{\zeta
\widetilde{K}^{+}_{H}(x) - imx}e^{-itx}dx,
$$
where $\widetilde{K}^{+}_{H}(x) = |x|^{1 - 2H}\Gamma(2H + 1)(\sin(H \pi) - i \textnormal{sign}(x) \cos(H \pi) )$. To obtain \eqref{p:V_gle_spec_repres}, rewrite the integrand \eqref{e:V_gle_def} as a fractional integral as in Pipiras and Taqqu \cite{pipiras:taqqu:2000} or, equivalently, adapt the expression for the spectral density of $V$ developed in Kou \cite{kou2008sat}, p.\ 524. We arrive at a Fourier domain integral with respect to a measure \eqref{e:Brownian_measure}. The spectral filter is
$$\frac{1}{\zeta \kappa(x) - im \hspace{0.5mm}\textnormal{sign}(x)|x|^{1+2d}} |x|^{d},$$
up to a constant; this leads to the filter in \eqref{p:V_gle_spec_repres} by elimination of the imaginary part.

Expression \eqref{e:fracOU_spec} shows that the fractional parametrization $\delta$ of the fOU encompasses the full range in \eqref{e:specdens_powerlaw}, whereas the fGLE is necessarily antipersistent. For the sake of illustration, the correlation structure of the fGLE is depicted in Figure \ref{f:fGLE_correl_structure}. Its autocovariance function displays the characteristic fluctuations associated with antipersistence. The spectral densities of the OU and fOU processes are displayed in Figure \ref{f:fOU_OU_specdens}; the latter has the characteristic singularity at the origin associated with long range dependence. To view qualitative differences between the processes, sample paths of the OU, fOU, and fGLE processes are shown in Figure \ref{f:OU_fracOU_samplepaths}.

\begin{figure}
	\begin{center}
	\includegraphics[height=2in,width=2.5in]{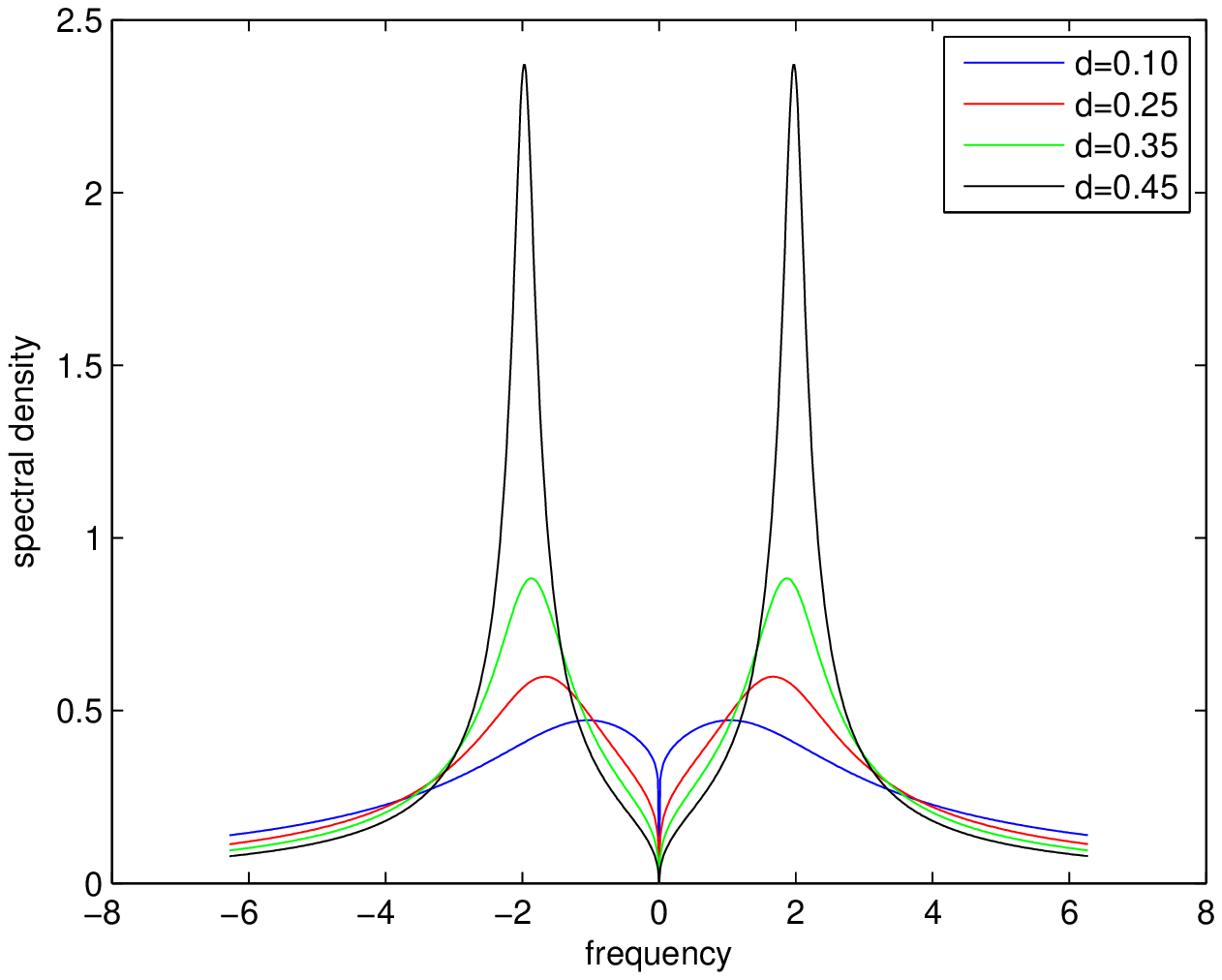} \ \includegraphics[height=2in,width=2.5in]{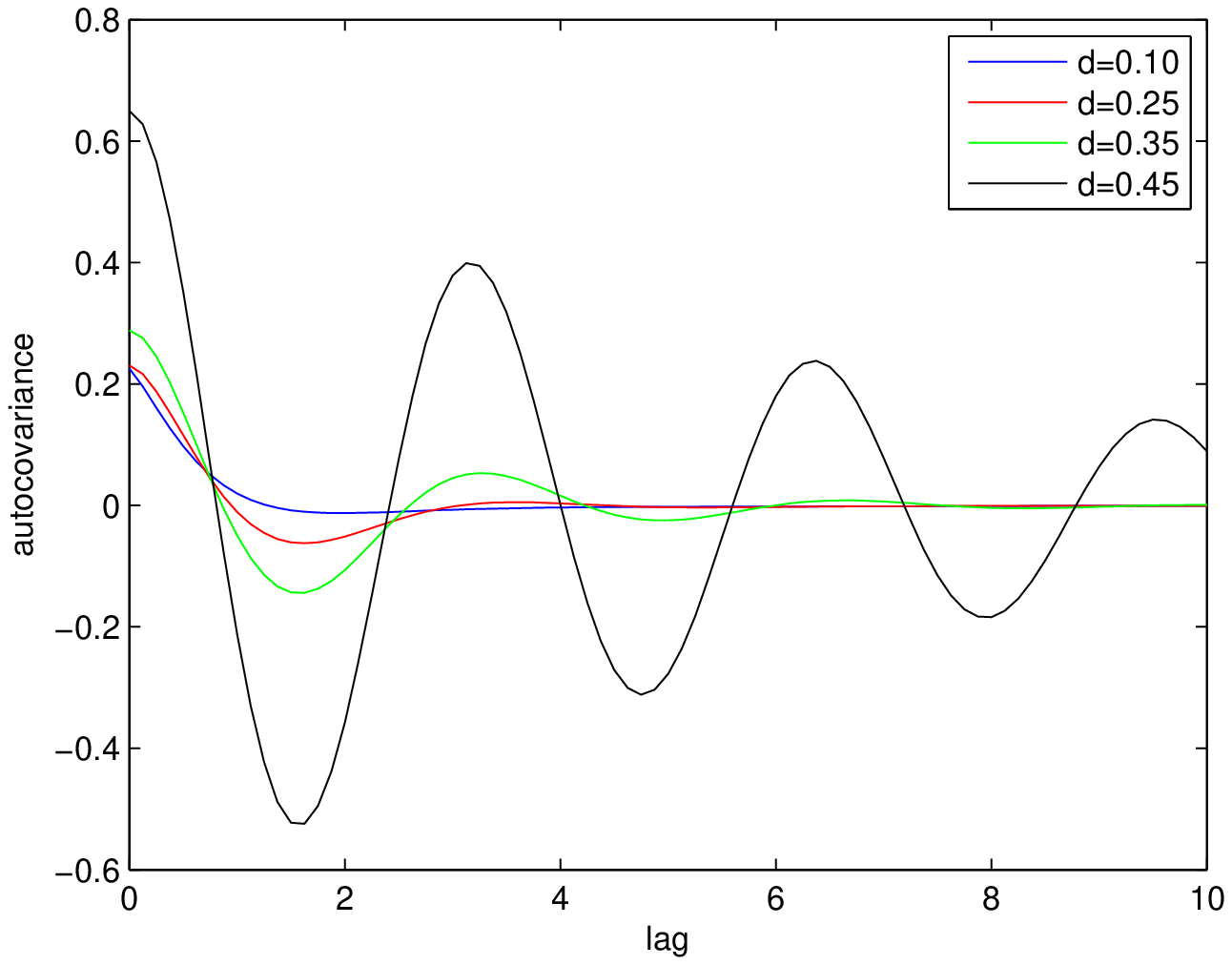}
	\caption{\label{f:fGLE_correl_structure} fGLE, correlation structure up to a multiplicative constant ($d =0.10,0.25,0.35,0.45$, $\zeta = 2$, $ m =1$). Left plot: spectral density. Right plot: autocovariance function.
	}
	\end{center}
\end{figure}

\begin{figure}
	\begin{center}
	\includegraphics[height=2in,width=2.5in]{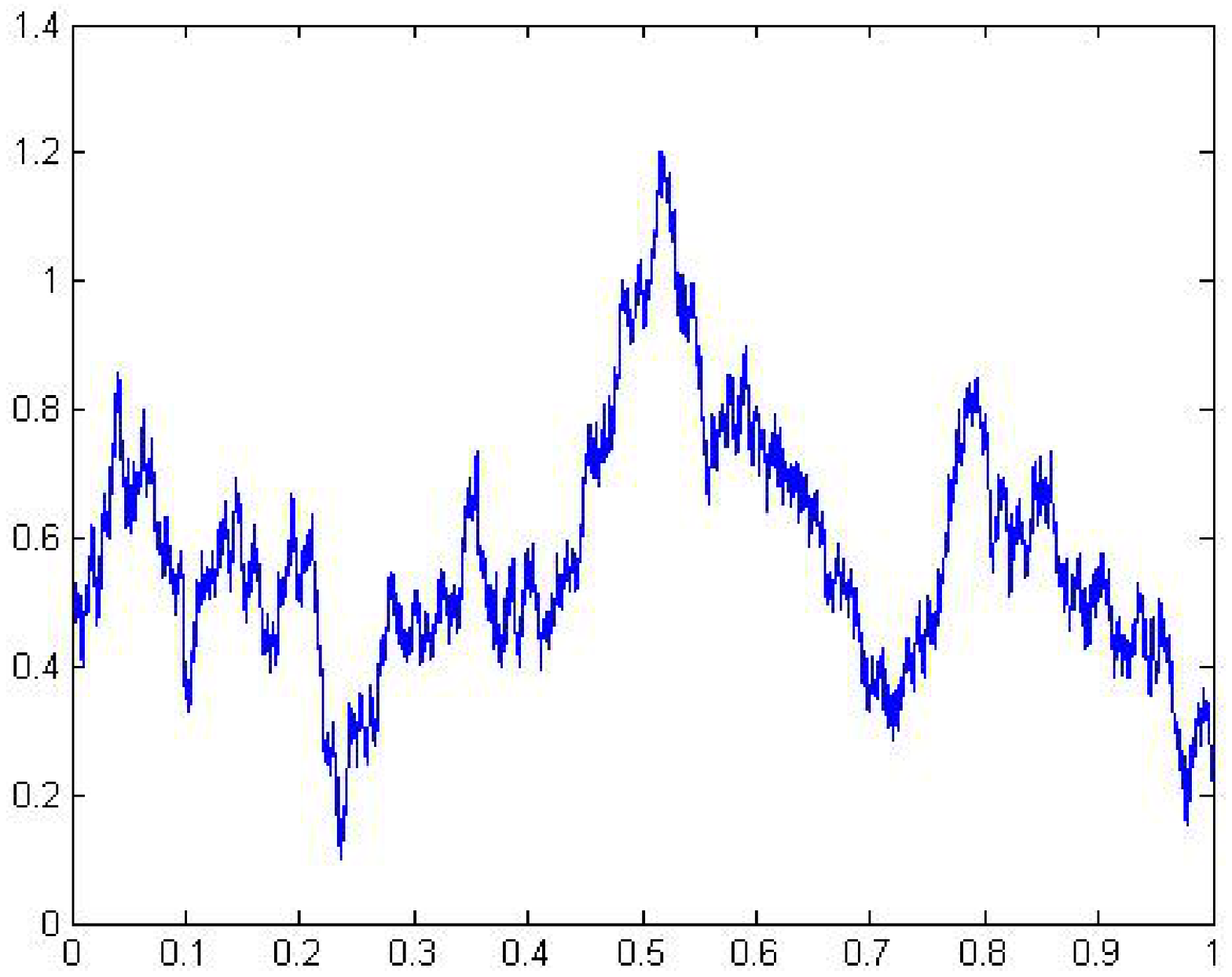} \ \includegraphics[height=2in,width=2.5in]{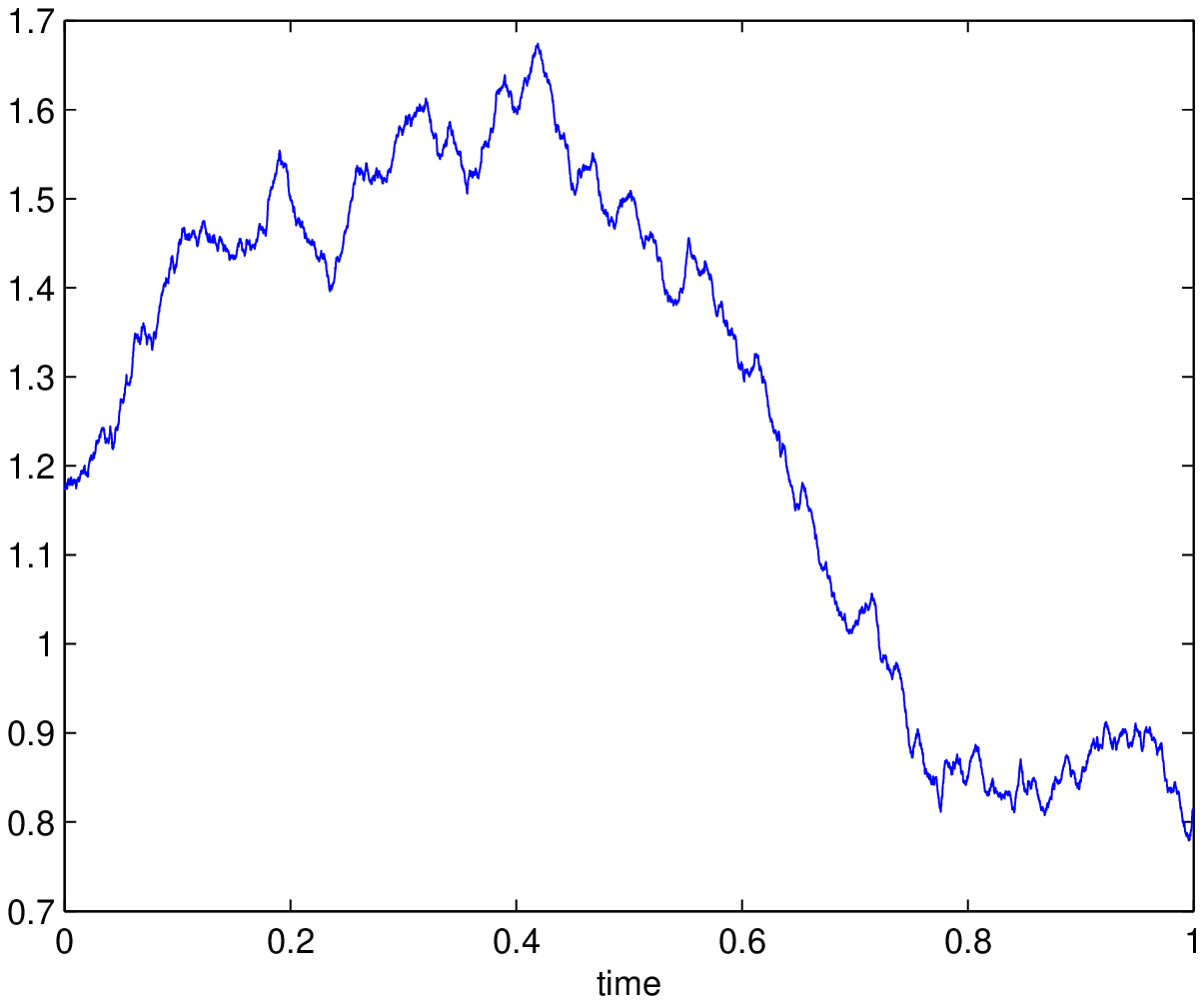}
\includegraphics[height=2in,width=2.5in]{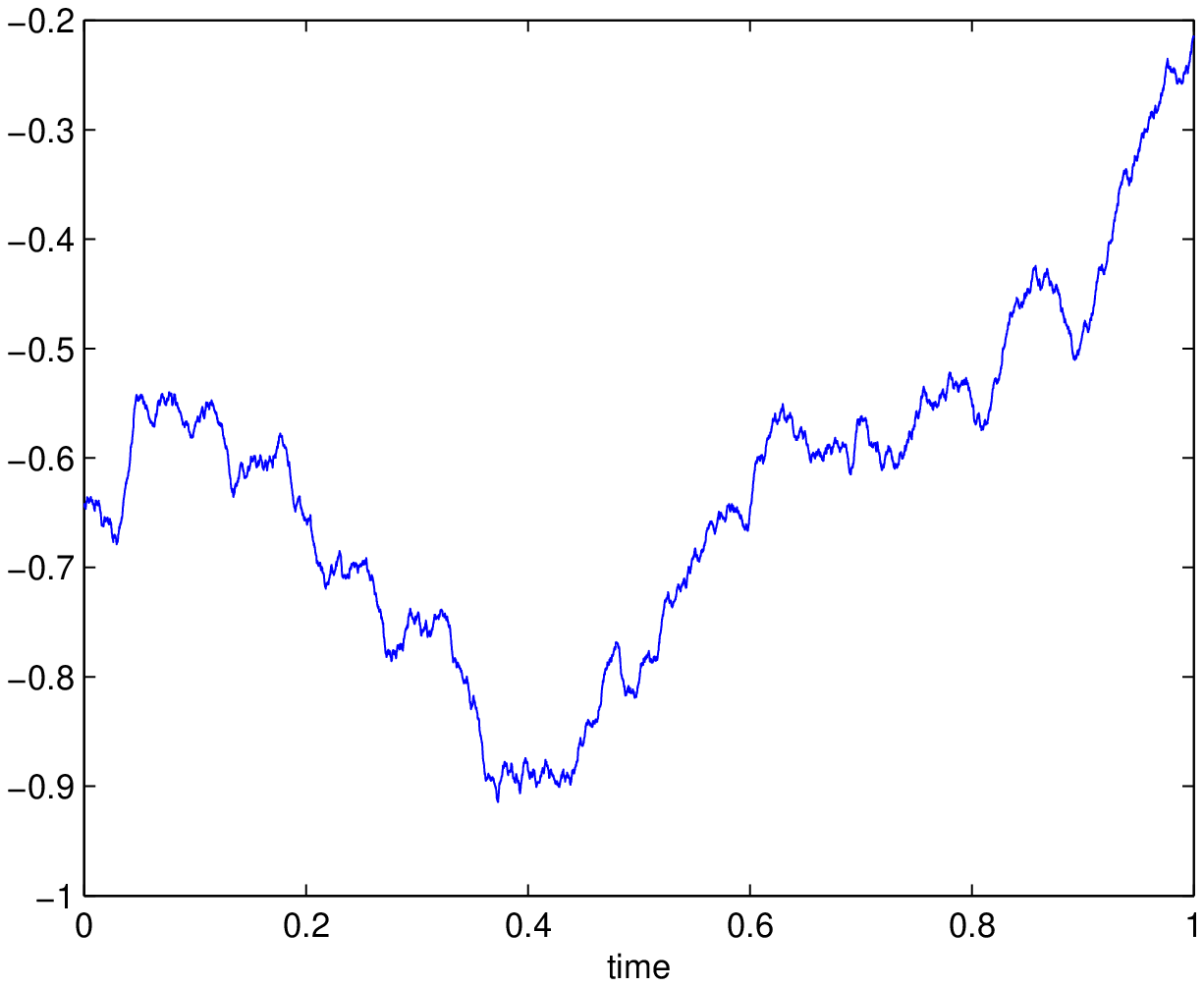}
	\caption{\label{f:OU_fracOU_samplepaths} Sample paths. Left: OU process, parameter value $\zeta =1$. Right: fOU process, parameter values $\zeta = 1$, $\delta =0.25$. Bottom: fGLE process, parameter values $\zeta = 2$, $m = 1$, $\delta = - 0.25$ (for more on the parametrization, see Section \ref{s:Fourier_integ}).
	}
	\end{center}
\end{figure}

\begin{figure}
	\begin{center}
	\includegraphics[height=2in,width=2.5in]{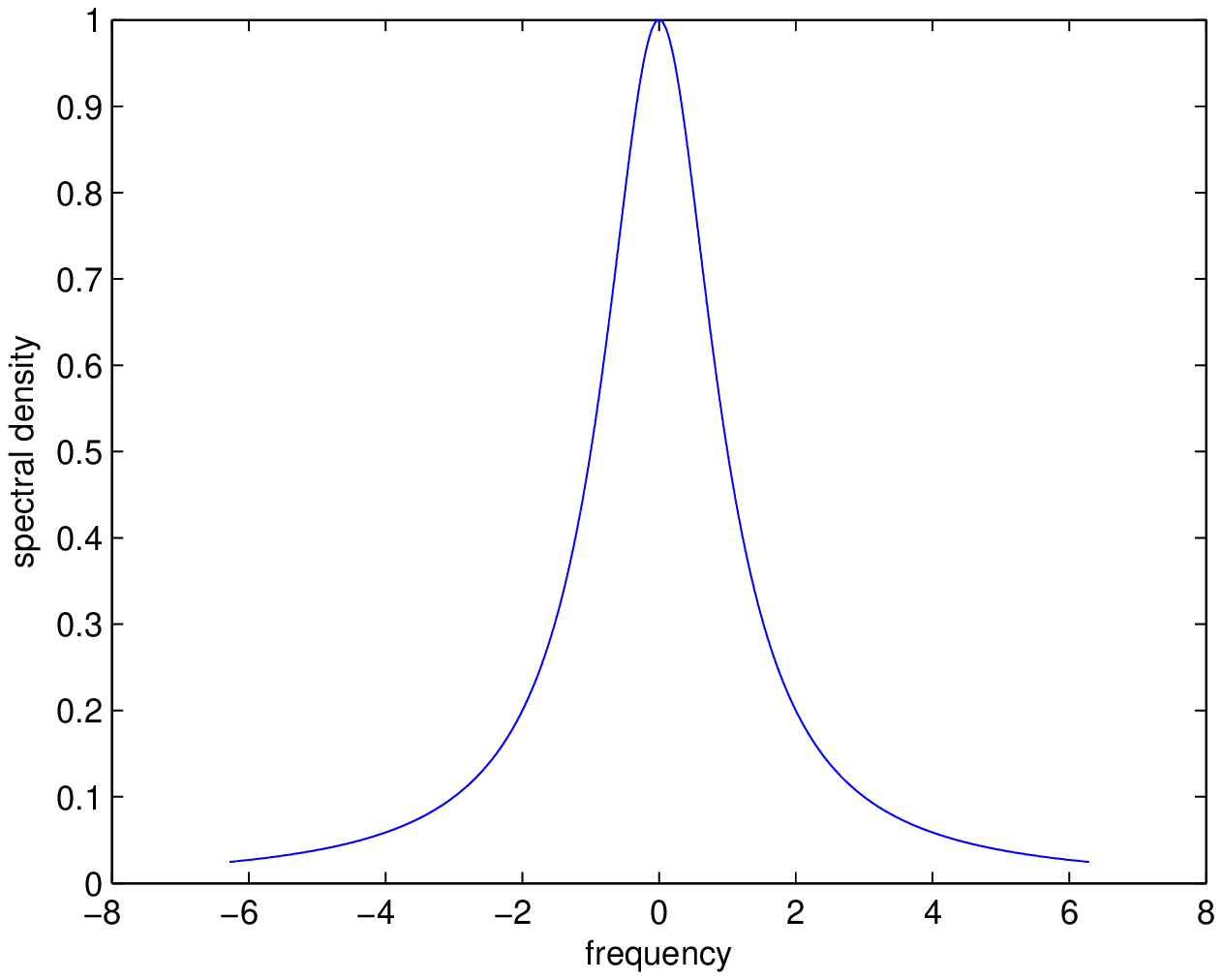} \ \includegraphics[height=2in,width=2.5in]{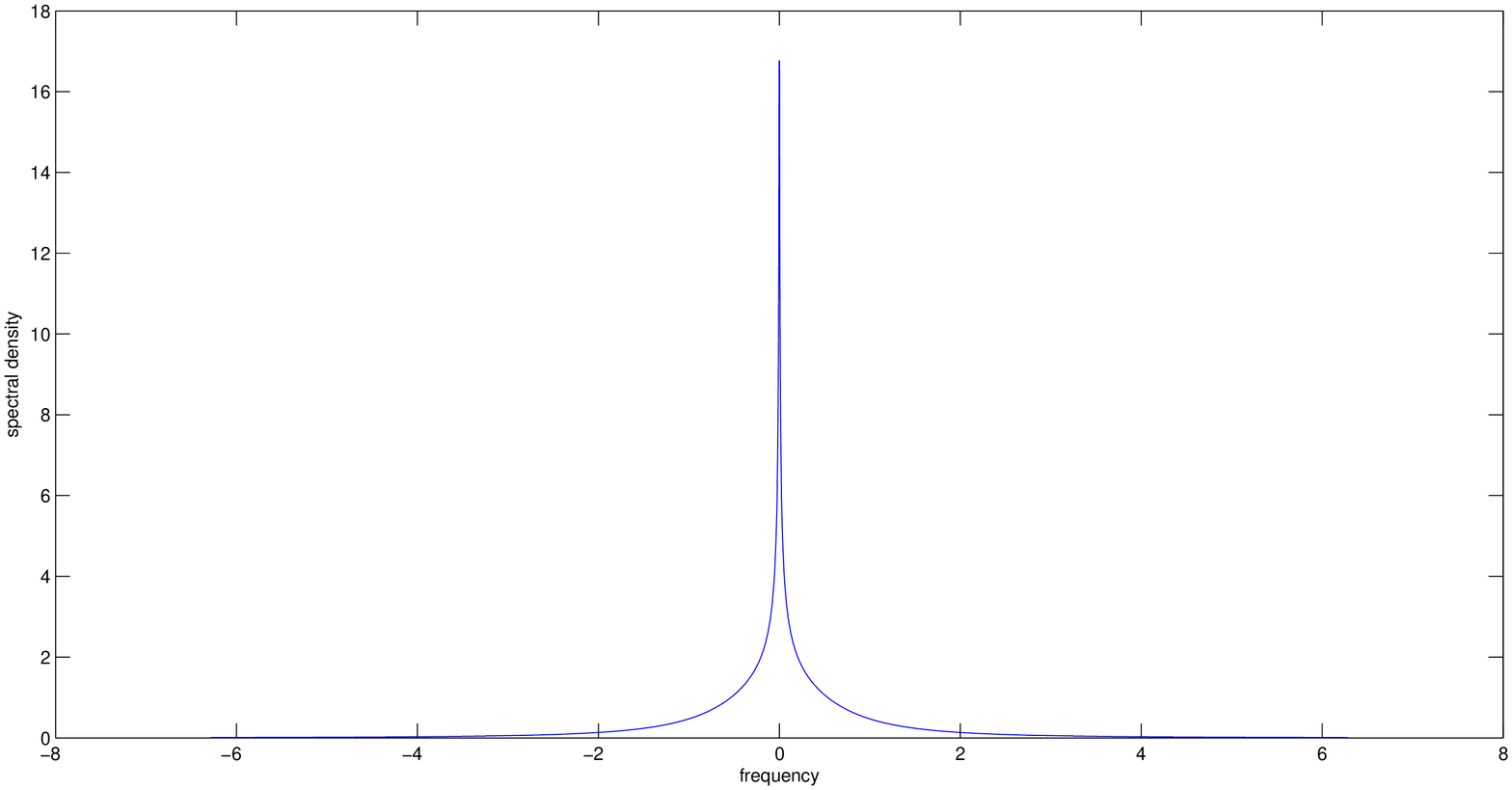}
	\caption{\label{f:fOU_OU_specdens} Spectral densities. Left: OU process, parameter value $\zeta =1$. Right: fOU process, parameter values $\zeta = 1$, $\delta =0.25$. (for more on the parametrization, see Section \ref{s:Fourier_integ}).
	}
	\end{center}
\end{figure}

\subsection{The simulation algorithm, wavelet decompositions and the choice of a discretization filter $g_j$}\label{s:choice}

In this subsection, we will present the wavelet-based simulation method and summarize previous results relevant to the current situation.  In addition, we will present the computational details of our proposed implementation for the simulation procedure.  Before proceeding to details, we will present sufficient conditions for the underlying wavelet decomposition to hold and give intuitive explanations for these assumptions.

The wavelet-based simulation procedure can be described as follows.
{\small
\begin{center}
\centering 
\begin{tabular}{|l|}
\hline \label{table:schematic_description_wavealg}
\\
Schematic description of wavelet-based simulation   \\ [0.5ex]
\\
\hline
\\
\textbf{Initialization}: generate via an exact method (e.g., Circulant Matrix
Embedding) \\ one first discretization sequence $V_0$;\\
\\
\textbf{Scale/step $j \in \bbN$}: given a discretization sequence $V_j$, obtain the next discrete\\
discretization $V_{j+1}$ at scale/step $j+1$ via the relation\\
\\
$V_{j+1} = u_j \ast \uparrow_2 V_j + v_j \ast \uparrow_2 \varepsilon_j$,\\
\\
where $\{\varepsilon_j\}$ is a noise sequence, and $(\uparrow_2 V_j)_k = V_{j,k/2}1_{\{\textnormal{even $k$}\}}$
is the upsampling \\
by factor 2 operator.
\\ [1ex]
\hline
\end{tabular}
\end{center}
}
\noindent The wavelet filters $u_j$ and $v_j$ are defined in the Fourier domain as
\begin{equation}\label{e:uj,vj}
\widehat{u}_j(x) = \frac{\widehat{g}_{j+1}(x)}{\widehat{g}_{j}(2x)}\widehat{u}(x), \quad \widehat{v}_j(x) = \widehat{g}_{j+1}(x)\widehat{v}(x),
\end{equation}
where $u$, $v$ are the conjugate mirror filters (CMF) of an underlying wavelet multiresolution analysis (MRA). For the theoretical purposes of this paper, we use a Meyer MRA, due to the compact Fourier domain support of both the scaling and wavelet functions (see Mallat \cite{mallat:1999}, chapter 7).

The algorithm works because at step $j+1$, the FWT annihilates the correlation structure $\widehat{g}_{j}$ of the discretization $\{V_{j,k}\}$ at scale $j$ (see expression \eqref{e:V_j_Wold}) and replaces it with a new, pre-chosen correlation structure $\widehat{g}_{j+1}$. The properties of the CMFs $u$ and $v$ play an important role, which we can explain heuristically. By taking Fourier transforms on both sides of \eqref{e:V_j_Wold},
$\widehat{V}_{j}(x) = \widehat{g}_{j}(x) \widehat{\xi}(x)$, where $\xi$ is a white noise sequence and $\widehat{\xi}(x)$ is its Fourier transform. Then
$$
\widehat{V}_{j+1}(x) = \frac{\widehat{g}_{j+1}(x)}{g_{j}(2x)} \widehat{u}(x) \widehat{V}_{j}(2x) + \widehat{g}_{j+1}(x) \widehat{v}(x) \widehat{\varepsilon}(2x),
$$
where $\varepsilon$ is another independent white noise sequence. Therefore, $ \widehat{V}_{j+1}(x) = \widehat{g}_{j+1}(x) (\widehat{u}(x) \widehat{\xi}(2x) + \widehat{v}(x) \widehat{\varepsilon}(2x))$. Since $u$ and $v$ are CMFs, then the term $\widehat{u}(x) \widehat{\xi}(2x) + \widehat{v}(x) \widehat{\varepsilon}(2x)$ is itself distributed as white noise. Thus, in law, $\widehat{V}_{j+1}(x) = \widehat{g}_{j+1}(x) \widehat{\eta}(x)$,
where $\eta$ is another white noise sequence; see also Pipiras \cite{pipiras:2005} for details.

In order to construct a wavelet-based decomposition of the velocity process $V$, the following technical assumptions must be met (see Didier and Pipiras \cite{didier:pipiras:2008}).\\

\noindent {\sc Assumption 1:}
$$
    \widehat g^{-1} \in L^2_{loc}(\bbR).
$$

\medskip
\noindent {\sc Assumption 2:} for any $j\in\bbZ$,
$$
    G_j,\ G_j^{-1}\in L^2_{loc}(\bbR).
$$

\medskip
\noindent {\sc Assumption 3:} for any $j_0\in\bbZ$,
$$
\max_{p = -1,1} \max_{k = 0,1,2} \sup_{j\geq j_0} \sup_{|x|\leq 4\pi/3} \left|
\frac{\partial^k (G_j(x))^p}{\partial x^k} \right| < \infty.
$$

\medskip
\noindent {\sc Assumption 4:} for large $|x|$,
$$
\left| \frac{\partial^k \widehat g(x)}{\partial x^k}\right| \leq
\frac{\mbox{const}}{|x|^{k+1}}, \quad k=0,1,2.
$$

\medskip
\noindent {\sc Assumption 5:} for large $j$,
$$
   | G_j(0) - 1 | \leq \textnormal{const} \hspace{1mm} 2^{-j}.\\
$$

\medskip

Assumptions 2, 3 and 5 pertain the choice of the discretization filter $\widehat{g}_{j}$. Assumption 3 is the rigorous version of the intuitive expression \eqref{e:GJ_approx_1} under a Meyer MRA. Assumption 5 states that the discrete and continuous time filters are arbitrarily close at frequency zero and is used to establish the main convergence result of this paper (Theorem \ref{t:VJ_conv_unif}). The imposition of assumptions on the inverses $\widehat{g}^{-1}$, $G^{-1}_j$ is related to the use of a biorthogonal wavelet basis (see Didier and Pipiras \cite{didier:pipiras:2008} for more details). The local square integrability in Assumptions 1 and 2 reflects the $L^2$ nature of the wavelet analysis coupled with the compact Fourier domain support of the underlying wavelet basis. Assumption 1 pertains only to $\widehat{g}$ and is clearly satisfied by \eqref{e:fracOU_spec}, \eqref{p:V_gle_spec_repres}. In order to clarify and explicitly incorporate processes with fractional spectral densities, Assumption 4 is replaced in Appendix \ref{s:adaptation_proofs} by the slightly modified Assumption $4^{'}$, which is used to establish the decay of the underlying modified wavelet basis via integration by parts.

As explained in Section \ref{s:intro}, the discretization of continuous time processes will in general induce discrete time processes whose spectral densities have intricate analytical forms. So, a natural question is whether one could construct converging discretizations by means of a simple method that applies to a wider class of stochastic processes, in particular, the fGLE and the fOU. Indeed, this can be done by developing (non-causal) discretization filters $\widehat{g}_{j}$ in three elementary steps:
\begin{itemize}
\item [($t.1$)] extend the truncated function $\widehat{g}(x)1_{[-\pi,\pi)}$ periodically to $\bbR$  implying that $\widehat{g}_j$ stems directly from $\widehat{g}$;
\item [($t.2$)] modify the resulting function with rescaling terms (e.g., $2^{j}$) so that relation (\ref{e:GJ_approx_1}) holds;
\item [($t.3$)] smooth the resulting function at $-\pi$ and $\pi$  to speed up the time domain decay of the filter in theory and computational practice.
\end{itemize}
Step $(t.1)$ replaces exact discretization filters such as those for the AR(1). Step $(t.2)$ is necessary to ensure the coherence between the discretization and the limiting process. Step $(t.3)$ minimizes the border effect, the consequence of the truncation of infinite-length filters which is always present in computational practice.  The method described in steps $(t.1) - (t.3)$ produces a sequence of discrete time processes $V_{j} = g_{j}\ast \varepsilon$ which is \textit{approximate}, since $g_j$ is selected for analytical and computational convenience. However, the choice of $g_j$ must still be such that the sequences $V_j$ convergence in an appropriate sense to the continuous time process.

%
%

We first look at the truncated procedure and filters obtained via steps $(t.1)$ and $(t.2)$. The idea of constructing $\widehat{g}_j(x)$ by truncating $\widehat{g}(x)$ at $\pm \pi$ and extending the function periodically to $\bbR$ has the obvious advantage of being a simple method for obtaining a discretization. Table \ref{t:truncated_filters} contains the proposed truncated filters generated based on components of the spectral filters for the OU, fOU and fGLE processes. The subscript $p$ denotes periodic extension beyond the domain $[-\pi,\pi)$. In all cases, we choose to deal with purely real filters for analytical simplicity.

\begin{table}[h]
\caption{\label{t:truncated_filters} Spectral density components and associated truncated filters (up to a constant)}
\centering
\begin{tabular}{ccccc}\hline
process & $|\widehat{g}(x)|^2$ & & &$\widehat{g}_{j}(x)$ \\ \hline
OU, fOU & $|\zeta^2 + x^2|^{-1}$  & & & $(\zeta^2 + 2^{2j}x^2)^{-1/2}_p$ \\
fOU, fGLE & $|x|^{-2\delta}$ & & & $2^{-j\delta}|x|^{-\delta}_p$ \\
fGLE & $|\gamma_0 + \gamma_1 |x|^{\beta}+ \gamma_2 |x|^{2\beta}|^{-1}$  & & & $ (\gamma_0 + \gamma_1 2^{j \beta} |x|^{\beta}+ \gamma_2 2^{2j \beta}|x|^{2\beta})^{-1/2}_p $\\
\end{tabular}
\end{table}
The associated functions \eqref{e:GJ_approx_1} are
\begin{equation}\label{e:G_{j,a}}
G_{j,\zeta}(x) = \frac{(\zeta^2 + 2^{2j}x^2)^{-1/2}_p}{(\zeta^2 + (2^j x)^2)^{-1/2}},
\end{equation}
\begin{equation}\label{e:G_{j,d}}
G_{j,d}(x) = \frac{2^{jd} |x|^d_{p}}{|2^j x|^d} = \frac{|x|^d_p}{|x|^d},
\end{equation}
\begin{equation}\label{e:G_{j,gamma,d}}
G_{j,\gamma,d}(x) = \frac{(\gamma_0 + \gamma_1 2^{j\beta}|x|^{\beta}+ \gamma_2 2^{2j\beta}|x|^{2\beta})^{-1/2}_{p}}{(\gamma_0 + \gamma_1 |2^{j} x|^{\beta} + \gamma_2 |2^{j}x|^{2 \beta})^{-1/2}},
\end{equation}
for $x \in \bbR$.

Even though it will typically create filters with discontinuities at $\pm \pi$ in the derivatives of the function $\widehat{g}_j(x)$, we can show that the filters exhibit good theoretical decay under assumptions. Mathematically, this requires replacing Assumption 3 with a set of weaker conditions. This is precisely stated in the next proposition.
\begin{proposition}\label{p:decay_AWD_filters}
Let $u$, $v$ be Meyer CMFs such that $\widehat{u}(x) \in C^{2}[-\pi,\pi)$, and let $\widehat{u}_j$, $\widehat{v}_j$ be wavelet-based filters as in \eqref{e:uj,vj}. Then, under Assumptions $3'$ and $4'$ (see Appendix \ref{s:adaptation_proofs}),
\begin{equation}\label{e:u,v_decay}
|u_{j,k}|, |v_{j,k}| \leq O\Big(\frac{1}{|k|^{2}} \Big).
\end{equation}
\end{proposition}
Furthermore, we can show that the approximate discretizations $\{V_{J,k}\}$ thus generated still converge exponentially fast to the correct limiting process. The two main mathematical results of this paper are found next. They state that the discretizations converge to the velocity and position processes a.s.\ uniformly over compact intervals (see also Remark \ref{r:smoothed_filters_ensure_conv} below).
\begin{theorem}\label{t:VJ_conv_unif}
Let $\{V(t)\}_{t \geq 0}$ be the velocity process for the fGLE \eqref{p:V_gle_spec_repres} or the fOU \eqref{e:fracOU_spec}, and let $\{V_{J,k}\}$ be discretization sequences generated according to the truncated filters in Table \ref{t:truncated_filters}. Then
$$
\sup_{t \in K}|2^{J/2}V_{J,\lfloor 2^{J}t \rfloor} - V(t)| \leq A 2^{-J \nu} \quad \textnormal{a.s.}
$$
for some $\nu \in (0,1)$, where $K$ is compact interval and $A$ is a random variable that does not depend on $J$.
\end{theorem}
\begin{corollary}\label{c:X_Riemann}
Let $\{V(t)\}_{t \geq 0}$ be the velocity process associated with the fGLE \eqref{p:V_gle_spec_repres} or the fOU \eqref{e:fracOU_spec}, and let $T \in \bbN$. Consider the subsequences $\{V_{J,k}\}_{k = 0,\hdots,2^{J}T-1}$ of the sequences in Theorem \ref{t:VJ_conv_unif}, where each term can be interpreted as
$V_{J,k} = V_{J,\lfloor2^J \frac{k}{2^J}\rfloor}$. Then
\begin{equation}\label{e:sup|X - sum V|}
\sup_{t \in [0,T]}\Big|X(t) - \sum^{\lfloor (2^J-1)t \rfloor}_{k=0}2^{J/2}V_{J,k}\frac{1}{2^J}\Big| \leq A' 2^{-J \nu}, \quad \nu \in (0,1),
\end{equation}
where the random variable $A'$ does not depend on $J$.
\end{corollary}
The proofs of Proposition \ref{p:decay_AWD_filters}, Theorem \ref{t:VJ_conv_unif} and Corollary \ref{c:X_Riemann} can be found in Appendix \ref{s:adaptation_proofs}.
\begin{remark}
Other approximation schemes with better computational error estimates can be attempted in Corollary \ref{c:X_Riemann}, such as a trapezoidal rule. We opted for an ordinary Riemann sum for simplicity and mathematical convenience.
\end{remark}
\begin{example}
In order to contrast the exact and approximate discretization filters we revisit the well-known case of the OU process.
The rationale behind the proposed $\widehat{g}_j(x)$ is simple. Multiplying the argument of a function by a number less than 1 has the geometric effect of ``stretching" the original function away from the origin. Therefore, we have that $\widehat{g}_{j}(2^{-j} x) =  (\zeta^2 + x^2)^{-1/2}$ (pointwise) for any fixed $x \in \bbR$ and large enough $j$, as desired. In other words, the discrete time process whose spectral density is $\widehat{g}_{j}(x)$ is, indeed, a discretization to the OU process in the sense explained in Section \ref{s:intro}. The difference between the truncated filter and the exact filters \eqref{e:GJ_OU_exact} can be seen in the Fourier domain in Figure \ref{f:OU_compare_filter}, left plot. As expected, even though $\widehat{g}_j(x)$ is continuous at $\pm \pi$, the same is not true for its first derivative. The left derivative of $\widehat{g}_j$ at $\pi$ is
$$
\widehat{g}^{'}_{j,-}(\pi) = - \frac{2^{2j}\pi}{(\zeta^2 + 2^{2j}\pi^2)^{3/2}}.
$$
Moreover, since $\widehat{g}_{j}(x) = 1/\sqrt{\zeta^2 + 2^{2j}(2 \pi - x)^2}$ for $x \in [\pi,2 \pi)$, then the right derivative of $\widehat{g}_j$ at $\pi$ is $\widehat{g}^{'}_{j,+}(\pi) = - \widehat{g}^{'}_{j,-}(\pi)$  by periodic extension.\\
\end{example}




\begin{figure}
	\begin{center}
	\includegraphics[height=2in,width=2.5in]{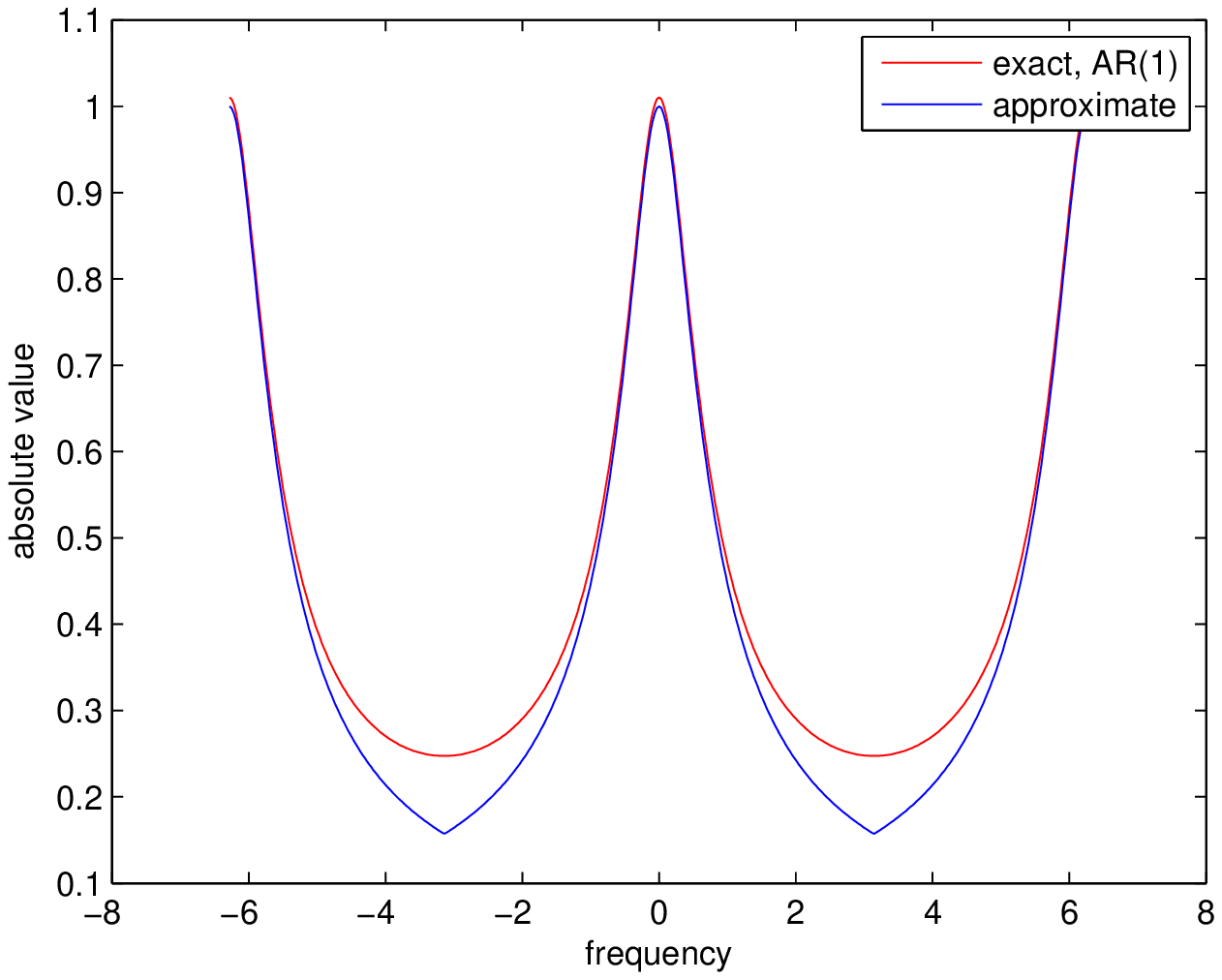} \ \includegraphics[height=2in,width=2.5in]{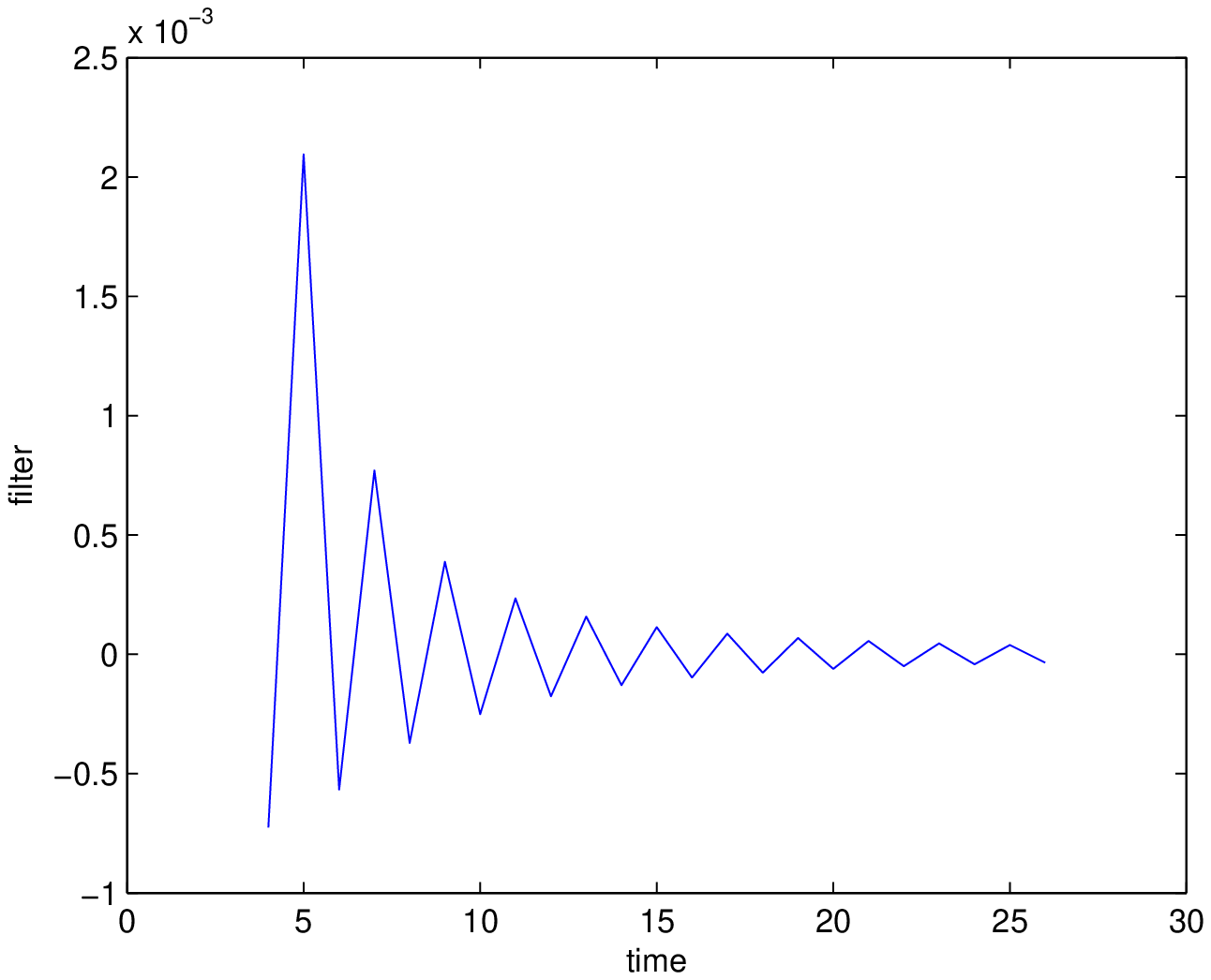}
	\caption{\label{f:OU_compare_filter} OU process, filters at $j=1$ ($\zeta = 1$, no zero moments).  Left: exact versus non-smoothed high-pass filter, spectral domain. Right: non-smoothed high-pass filter, time domain.
	}
	\end{center}
\end{figure}

Although the theoretical guarantees on the asymptotic decay of the truncated filters are good (Proposition \ref{p:decay_AWD_filters}), the discontinuities at $\pm \pi$ in the derivatives of the function $\widehat{g}_j(x)$ will create ripples in the time domain expression of the associated filters. This can be seen for the OU process filters in Figure \ref{f:OU_compare_filter}, right plot. In computational practice, one could ask whether it is possible to eliminate or at least to minimize the ripples observed in the plots. Moreover, and related to this, filters with fast time domain decay are desirable because in practice, infinite length filters used in convolution-based algorithms must be truncated (the border effect). One classical way to speed up the time domain decay of the designed filter is to generate a smoother Fourier domain expression. Step ($t.3$) consists of addressing this issue by smoothing the kinks of $\widehat{g}_j(x)$ at $\{(2k+1) \pi \}_{k \in \bbZ}$. There is clearly more than one way to do this. Table \ref{t:truncated_filters-smoothed} contains the proposed smoothed filters. Denote the discretization filters in the first through third rows of Table \ref{t:truncated_filters} by $\widehat{g}_{j,\zeta}(x)$, $\widehat{g}_{j,d}(x)$, $\widehat{g}_{j,\gamma,d}(x)$, respectively. Thus, the proposed filters for the OU, fOU and fGLE processes are $\widehat{g}_{j,\zeta}(x)$, $\widehat{g}_{j,\zeta}(x)\widehat{g}_{j,d}(x)$ and $\widehat{g}_{j,\zeta}(x)\widehat{g}_{j,\gamma,d}(x)$, respectively. In contrast with exact discretization filters (see Figure \ref{f:timedomain_truncpatch}), they generally give rise to non-causal low- and high-pass filters $u_j$ and $v_j$. Causality can be a desirable property for certain applications (for instance, see Didier and Pipiras \cite{didier:pipiras:2010}, section 6.2), but its absence has no effect on our simulation method.

\begin{table}[h]
\caption{\label{t:truncated_filters-smoothed} Spectral density components and associated truncated-smoothed filters (up to a constant)}
\centering
\begin{tabular}{lccc}\hline
process & $|\widehat{g}(x)|^2$ & $\widehat{g}_{j}(x)$ & $x^*(j) > 0$\\ \hline
OU, fOU & $|\zeta^2 + x^2|^{-1}$  & $\frac{\exp({\frac{\upsilon}{2 \pi^2}(\frac{x^{*}(j)x}{\pi})^2})}{(\zeta^2 + 2^{2j}(\frac{x^{*}(j)x}{\pi})^2)^{1/2}_p}$ & $\sqrt{\frac{\pi^2}{\upsilon} - \frac{\zeta^2}{2^{2j}}}$, \\
& & & where $\upsilon > 0$, \\
& & & $\frac{\pi^2}{\upsilon} > \frac{\zeta^2}{2}$ \\
 & & & \\
fOU, fGLE & $|x|^{-2\delta}$ & $\frac{ \exp(\frac{\textnormal{sign}(\delta)}{2 \pi^2} ( \frac{x^{*}(j)x}{\pi})^2 )}{2^{j \delta} | \frac{x^{*}(j)x}{\pi}|^{\delta}_p }$ & $\pi \sqrt{|\delta|}$\\
 & & & \\
fGLE & $|\gamma_0 + \gamma_1 |x|^{\beta}+ \gamma_2 |x|^{2\beta}|^{-1}$  & $\frac{\exp(\frac{\beta}{2 \pi^2} x^2 )}{(\gamma_0 + \gamma_1 2^{j\beta}|x|^{\beta}+ \gamma_2 2^{2j\beta}|x|^{2\beta})^{1/2}_{p}}$ & - \\
\end{tabular}
\end{table}

\begin{remark}\label{r:smoothed_filters_ensure_conv}
The difference between the non-smoothed, truncated filters and the smoothed ones rests solely on a multiplicative exponential factor and rescaling of the argument. Based on the calculations for the former (see Lemma \ref{l:Gj_satisfies_A2_A3'_A5}), it is thus easy to show that the smoothed filters in Table \ref{t:truncated_filters-smoothed} satisfy the conditions for convergence used in Theorem \ref{t:VJ_conv_unif} and Corollary \ref{c:X_Riemann}.
\end{remark}

\begin{example}
Multiplying the original truncated discretization filter for the OU process by a term of the form $e^{\frac{\upsilon x^2}{2 \pi^2}}$ creates in the former two global minima, symmetrically to the left and to the right of the origin, since the rapid growth of the exponential term eventually prevails over the decay to zero of the inverse polynomial. The parameter $\upsilon$ should only ensure that $x^*(j) \in \bbR$. By relocating these minima $x^*(j)$ to $\pm \pi$ via rescaling, we obtain periodic functions $ \widehat{g}_j \in C^{\infty}[-\pi,\pi)$, which decay in the time domain faster than any inverse polynomial. As a consequence, the high-pass filter $\widehat{v}_{j}(x) = \widehat{g}_{j}(x)\widehat{v}(x)$ is also quite smooth. The fact that some neighborhood $B(0,\delta)$ is not contained in $\textnormal{supp}(\widehat{v})$ implies that the near-spikes of $\widehat{g}_{j}(x)$ at $x = 0$, a potential source of ripples in the time domain filter, disappear in the inverse Fourier transform of $\widehat{v}_j$. Similar remarks apply to the low-pass filter $\widehat{u}_{j}(x)$. The graph of the smoothed filter for $j=2$ can be be seen in Figure \ref{f:spec_density_truncpatch}. The domain taken is $[-\pi,3\pi)$ to illustrate that the periodic extension to $\bbR$ is quite smooth.  Figure \ref{f:timedomain_truncpatch} displays the resulting time domain filter after numerical integration; moreover, for any fixed $x$, $\widehat{g}_{j}(2^{-j}x) \rightarrow \widehat{g}(x)$ as $j \rightarrow \infty$.
\end{example}

\begin{figure}
	\begin{center}
	\includegraphics[height=2in,width=2.5in]{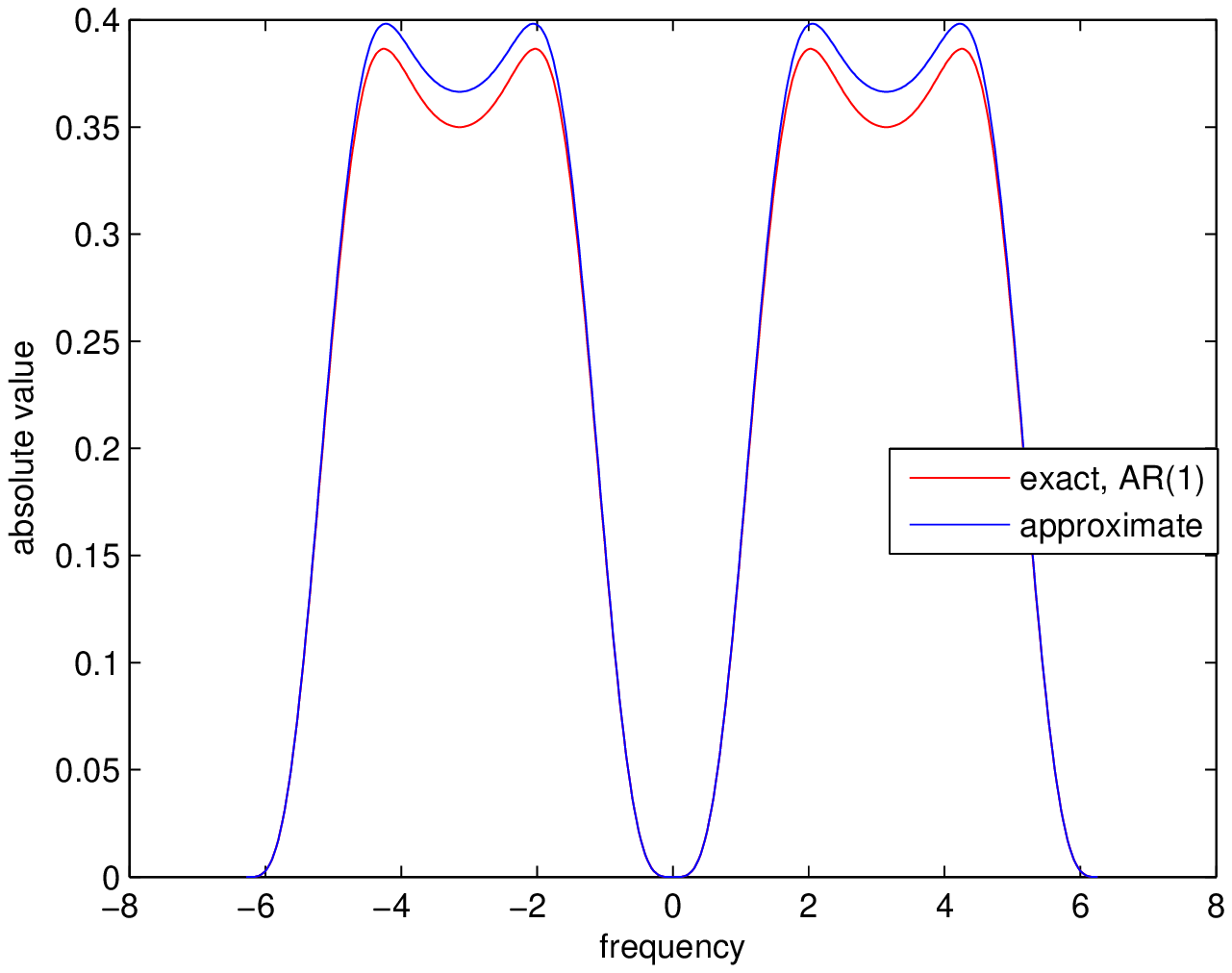} \ \includegraphics[height=2in,width=2.5in]{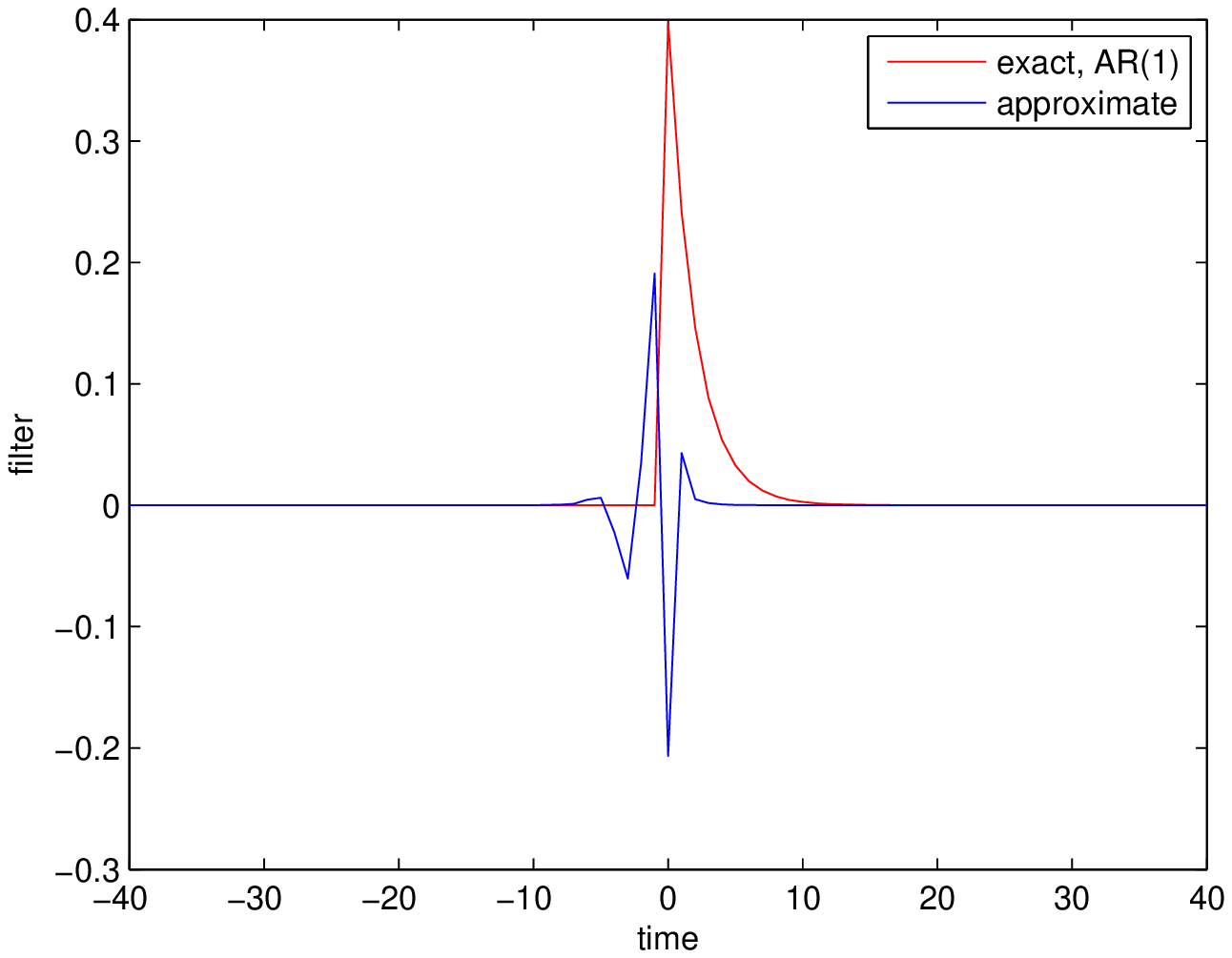}
	\caption{\label{f:timedomain_truncpatch} OU process, exact versus approximate high-pass filters at $j=1$ ($\zeta = 1$, 4 zero moments). Left: spectral domain. Right: time domain.
	}
	\end{center}
\end{figure}

\begin{figure}[t]
\begin{center}
	\includegraphics[height=2in,width=2.5in]{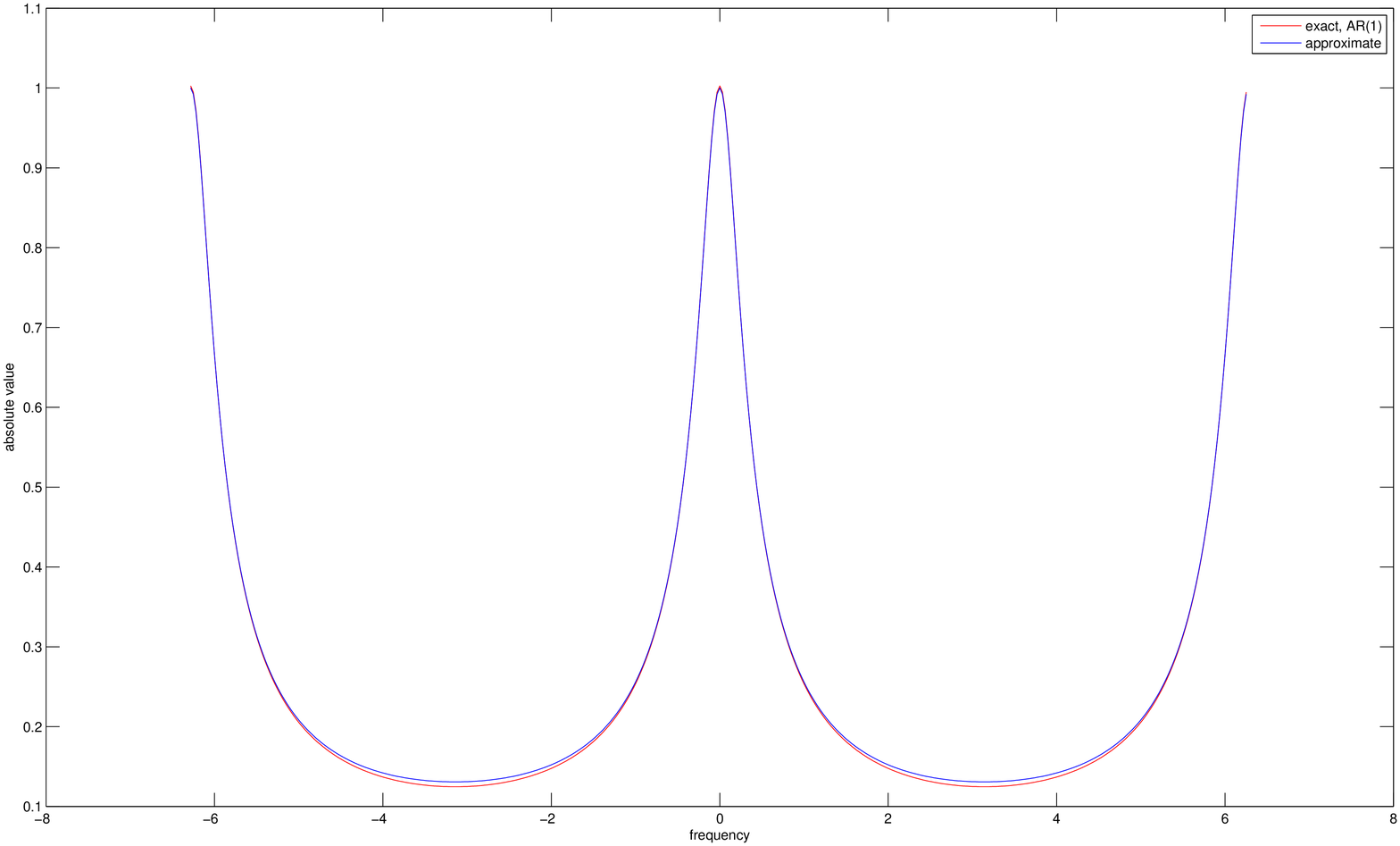} \ \includegraphics[height=2in,width=2.5in]{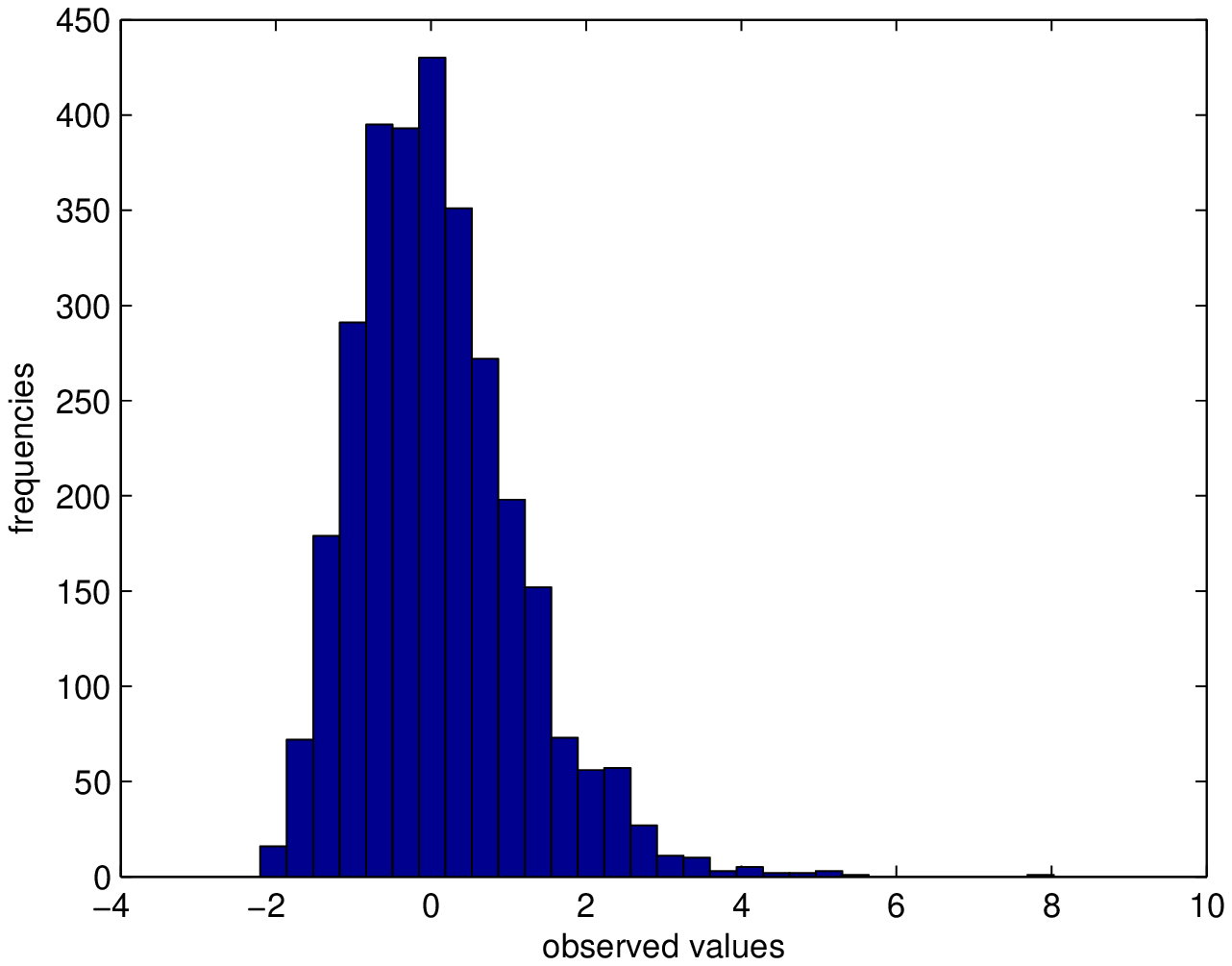}
	\caption{\label{f:spec_density_truncpatch} Left plot: OU process, spectral filter $\widehat{g}_j$ for $j=2$ ($\zeta = 1$). Right plot:
fOU process ($\zeta = 1$, $d = 0.25$, time series length $2^9$) histogram for $3000$ Monte Carlo runs of the statistic $\widetilde{T}$ as generated via the wavelet method.}
	\end{center}
\end{figure}

%

In the case of the fOU, the resulting filters $\widehat{g}_j$, $\widehat{v}_j$ are displayed in Figure \ref{f:fracOU_spec_density_truncpatch} for $d = 0.25$ and $j = 2$. The left and right plots illustrate the effect of multiplication by the high pass wavelet filter $\widehat{v}$. For $d > 0$, the filters $\widehat{g}_{j,d}$ and $\widehat{g}_{j}$ shows a singularity at the origin. This makes the role of the compact support high pass wavelet filter $\widehat{v}$ quite important for the numerical stability of the computation of the associated time domain filters. Analogously, the two plots in Figure \ref{f:fracOU_timedomain_truncpatch_low} illustrate the effect of the low pass wavelet filter $\widehat{u}$. Even before multiplying by $\widehat{u}$, there is no singularity at $x = 0$, because the singularities in the individual terms $\widehat{g}_{j}$, $\widehat{g}_{j-1}$ cancel out in the ratio $\widehat{g}_{j}(x)/\widehat{g}_{j-1}(2x)$. The resulting low and high pass filters $u_j$, $v_j$ in the time domain are shown in Figure \ref{f:fracOU_timedomain_truncpatch}. A generated sample path can be viewed in Figure \ref{f:OU_fracOU_samplepaths}. The remarkable persistence in the sample path, especially in comparison with that of the OU process, is due to the long range dependence of the fOU process when $d > 0$.

\begin{figure}
	\begin{center}
	\includegraphics[height=2in,width=2.5in]{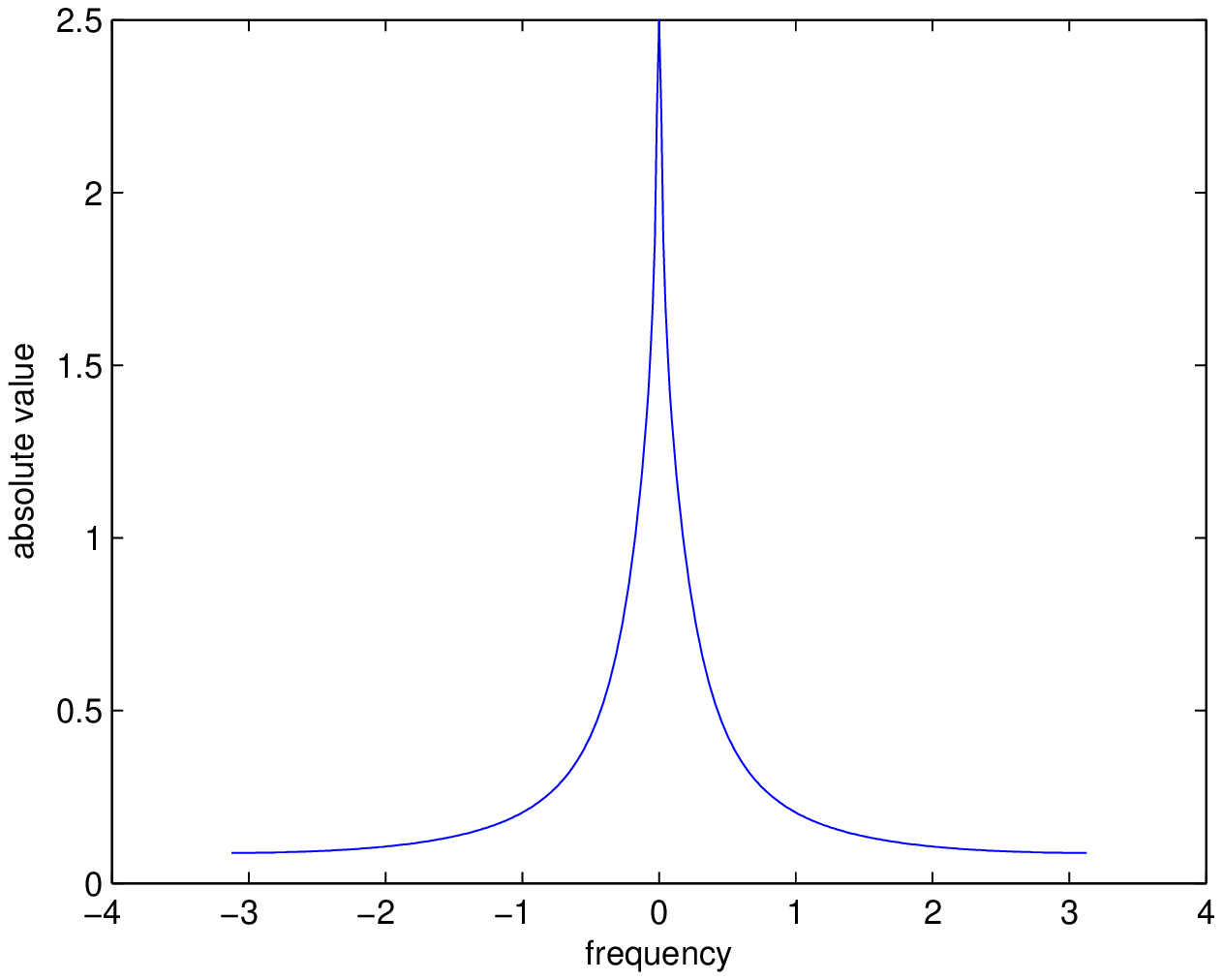} \ \includegraphics[height=2in,width=2.5in]{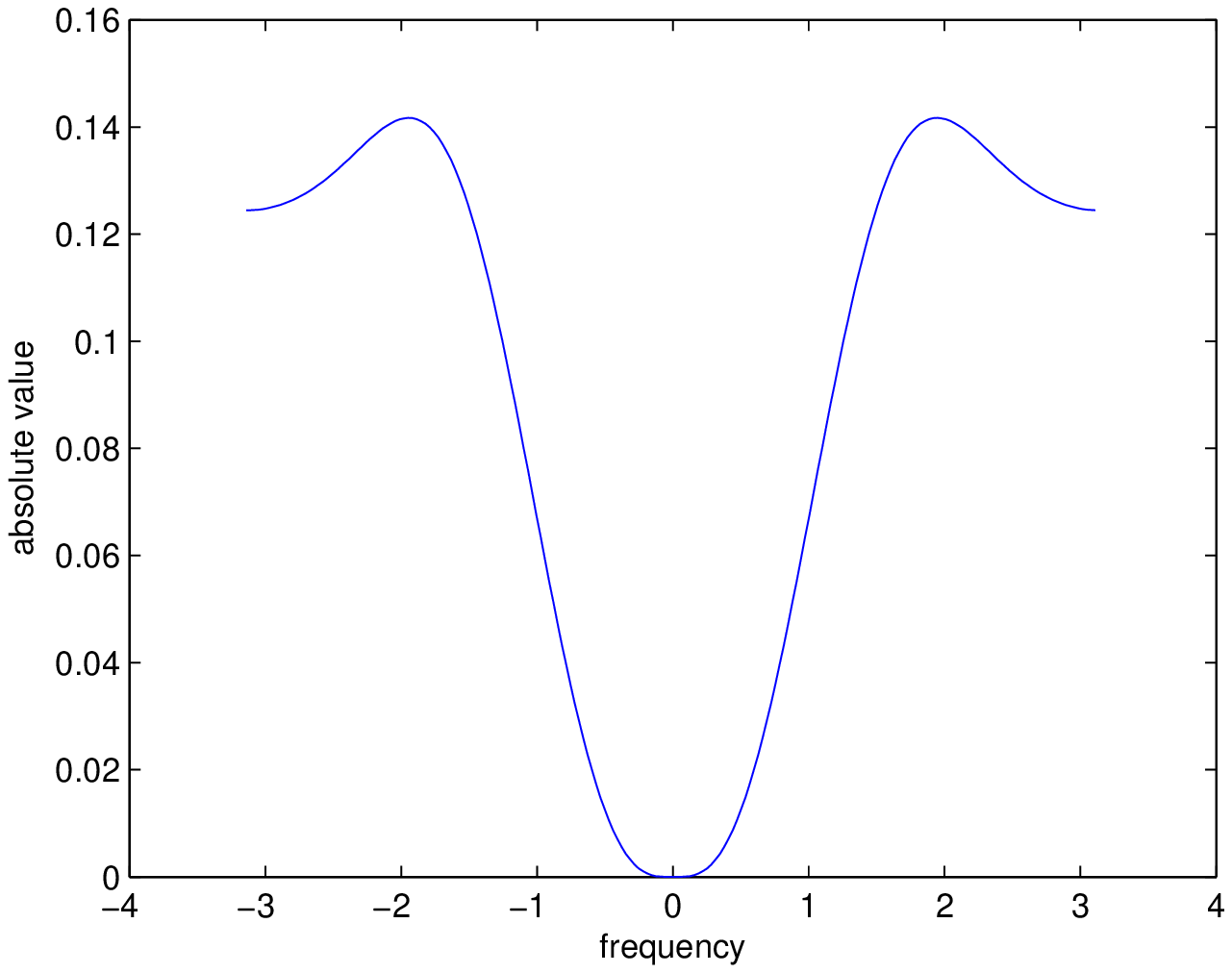}
	\caption{\label{f:fracOU_spec_density_truncpatch} fOU process, high pass spectral filters at $j=2$ ($\zeta = 1$, $d =0.25$). Left: $\widehat{g}_j$. Right: $\widehat{v}_j$ (4 zero moments).
	}
	\end{center}
\end{figure}

\begin{figure}
	\begin{center}
	\includegraphics[height=2in,width=2.5in]{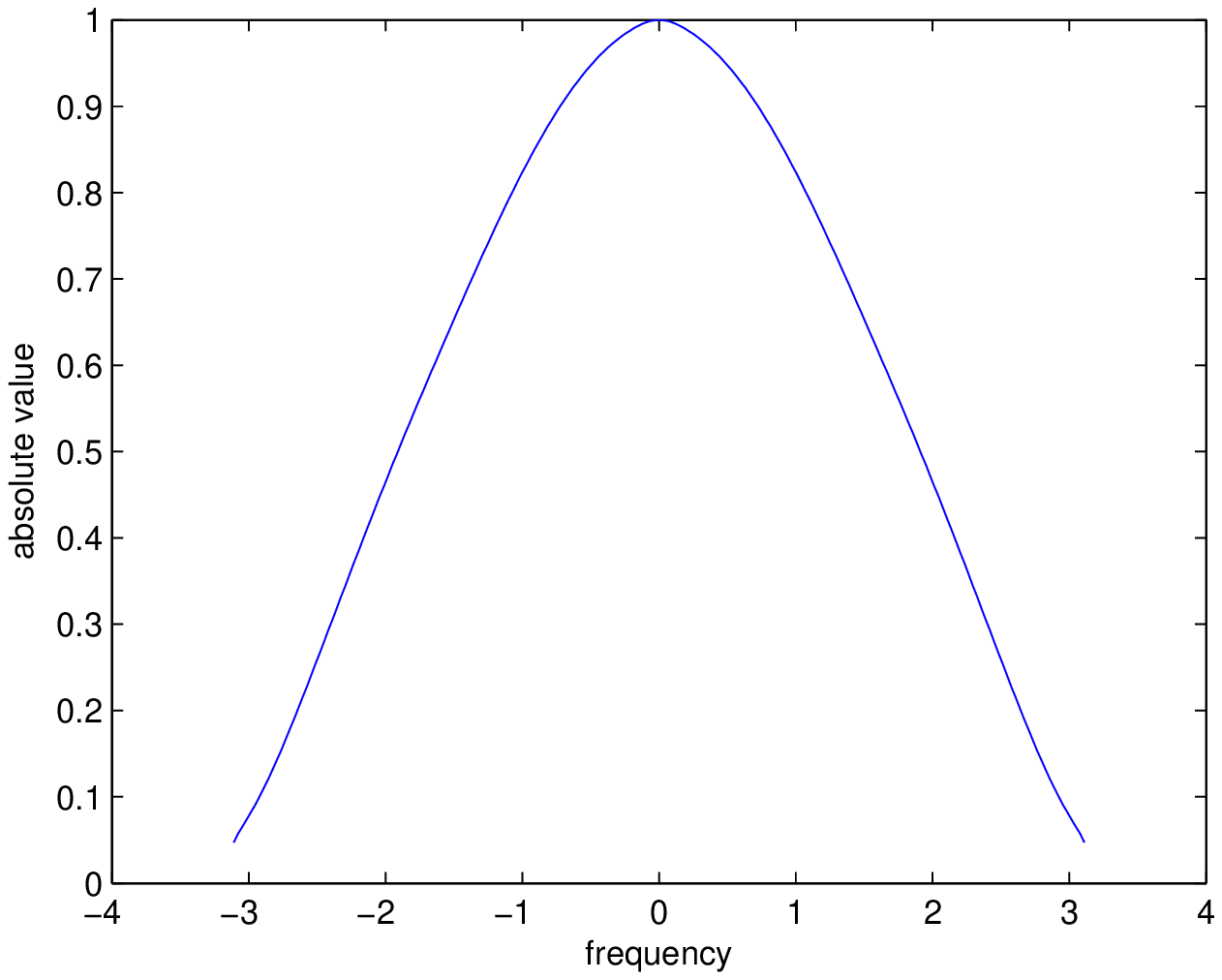} \ \includegraphics[height=2in,width=2.5in]{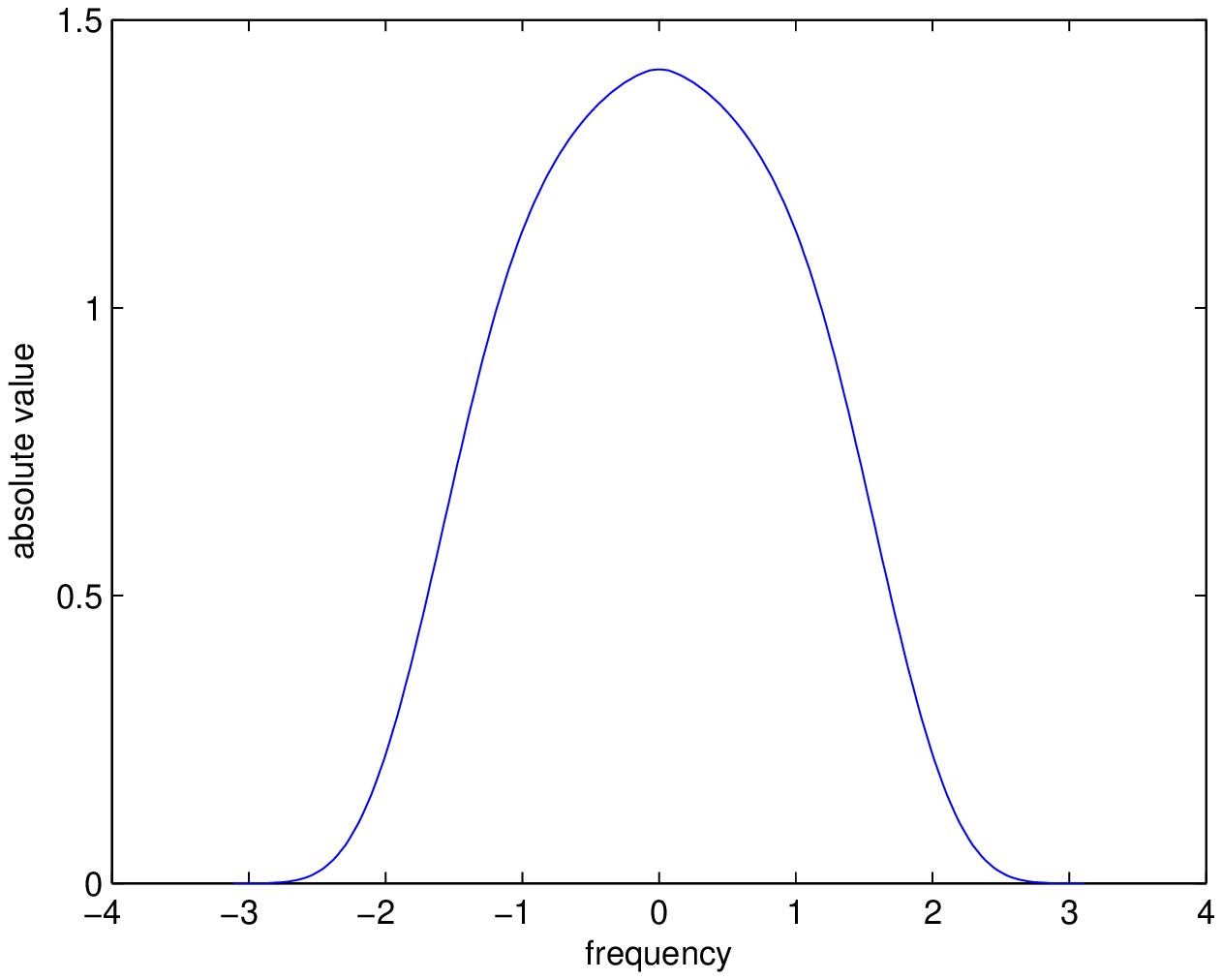}
	\caption{\label{f:fracOU_timedomain_truncpatch_low} fOU process, low pass spectral filters at $j=2$ ($\zeta = 1$, $d =0.25$). Left: $\widehat{g}_{j}(\cdot)/\widehat{g}_{j-1}(2\cdot)$. Right: $\widehat{u}_j$ (4 zero moments).
	}
	\end{center}
\end{figure}

\begin{figure}
	\begin{center}
	\includegraphics[height=2in,width=2.5in]{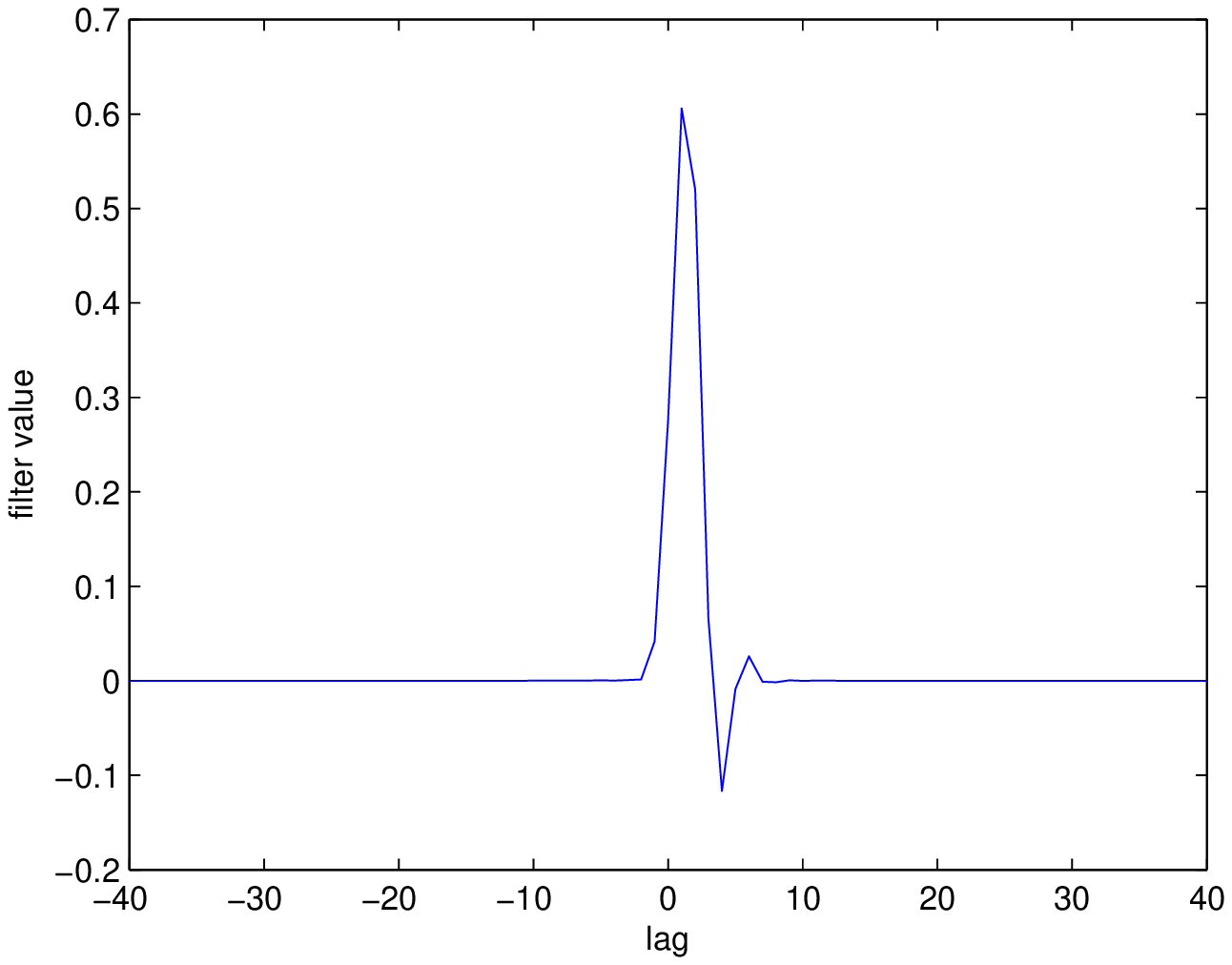} \ \includegraphics[height=2in,width=2.5in]{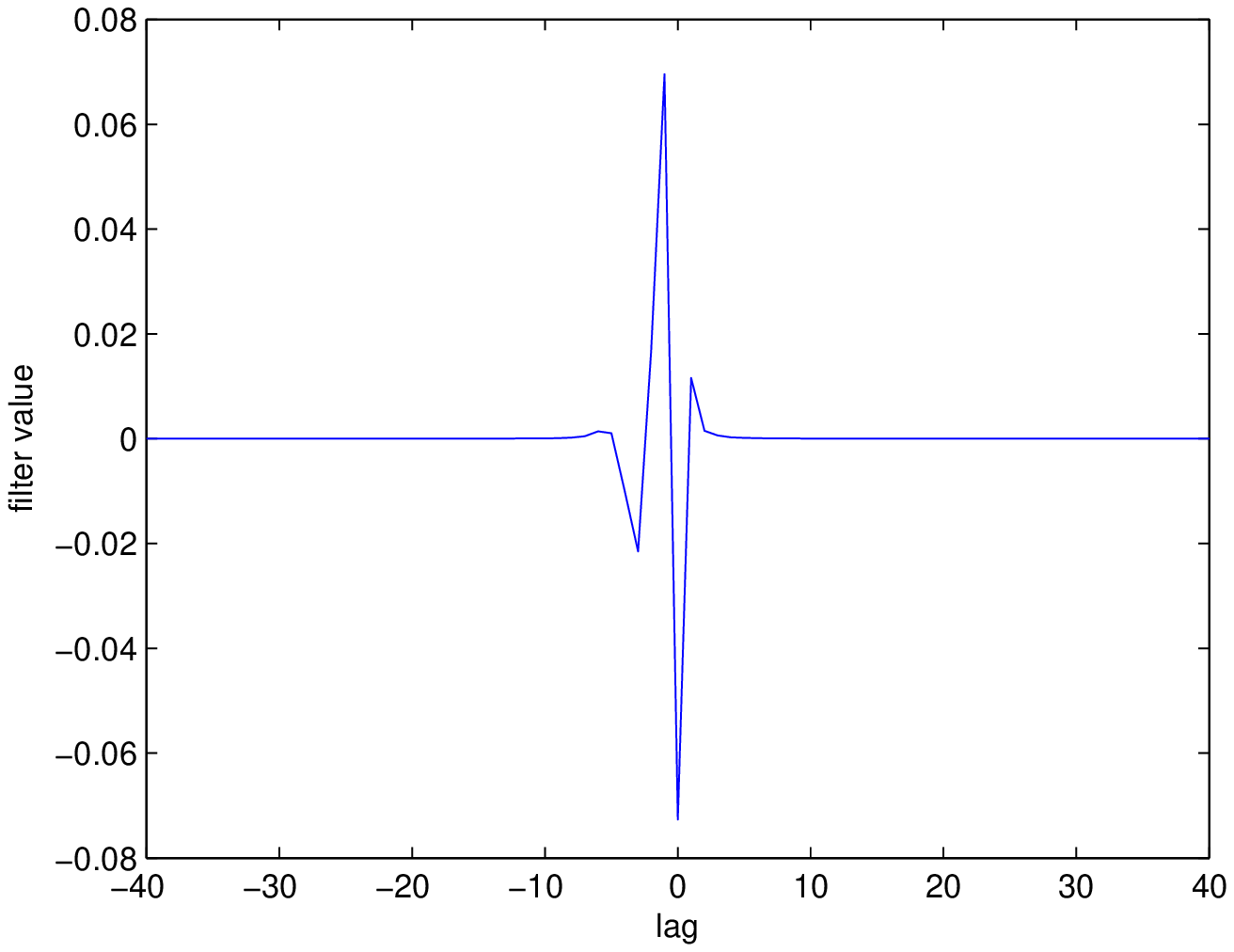}
	\caption{\label{f:fracOU_timedomain_truncpatch} fOU process, time domain filters at $j=2$ ($\zeta = 1$, $d =0.25$, 4 zero moments). Left: $u_j$. Right: $v_j$.
	}
	\end{center}
\end{figure}

As for the fGLE, smoothing is more challenging due to the presence of spikes in its spectral density (see Figure \ref{f:fGLE_correl_structure}).  Also, its more complicated functional form makes manipulation more difficult; however, we propose to continue to smooth via an exponential term due to the ensuing analytical simplicity. For large enough $j$, $\pm \pi$ become \textit{approximate }critical points of the component $\widehat{g}_{j,\gamma,d}(x)$, thus dispensing with centering. In fact, from Table \ref{t:truncated_filters-smoothed}, the first order condition for the logarithm of the proposed filter in the range $x > 0$ is that
\begin{equation}\label{e:FOC}
\beta \frac{\gamma_0}{2^{2j \beta}}x + \beta \frac{\gamma_1}{2^{j \beta}} x^{\beta+1} + \beta \gamma_2 x^{2\beta+1} =
\frac{\pi^2}{2} \frac{\beta \gamma_1 }{2^{j \beta}}x^{\beta-1} + \frac{\pi^2}{2} \gamma_2 2 \beta x^{2 \beta-1}.
\end{equation}
The first and second terms on the left-hand side, and the first term on the right-hand side of (\ref{e:FOC}) are close to zero for sufficiently large $j$. Thus, solving the resulting expression for $x$ gives $x^*(j) = x^* = \pi$ as an approximation. Moreover, the approximation \eqref{e:GJ_approx_1} holds as pointwise limit.


\begin{remark}
The simplification behind (\ref{e:FOC}) requires large enough $j$ so that the critical point is close to $\pi$. However, this depends on the values of the parameters $\zeta$ and $m$. For instance, the high pass filter $\widehat{u}_1(x)$ for $(\zeta,m) = (2,1)$ is already quite smooth. However, for $(\zeta,m) = (10,1)$, $\widehat{u}_1(x)$ displays a kink at $\pi/2$, which is already much less visible in $\widehat{u}_2(x)$ (not shown).
\end{remark}

\begin{remark}
As discussed in Pipiras \cite{pipiras:2005}, increasing the number of vanishing moments of the underlying MRA can improve the time domain decay of certain filters. Strictly speaking, the improvement in the decay depends on the specific Fourier domain form of the filter. (However, see Didier and Pipiras \cite{didier:pipiras:2010}, especially Remarks 6 and 7.)

In the context of the present paper, we can give a slightly different explanation for the potentially positive effect of the number of vanishing moments. Increasing the latter has the effect of increasing the regularity (Fourier domain smoothness) of the proposed filters. With a greater number of vanishing moments, $\widehat{v}_j(x)$ becomes flatter over a wider vicinity of zero in the Fourier domain, which can more efficiently make up for a singularity or kink at the origin (e.g., for the fOU). In numerical studies, we compared the time domain decay of fOU low- and high-pass filters $u_j$ and $v_j$ for $N = 4$ or 8 vanishing moments and parameters values $\zeta = 1$, $d = 0.25$. In general, for $j=2$, 5, 8 and 11 the tail values (lags $T = 31$ through 40) of $u_j$ and $v_j$ under $N = 8$ was on average of the order $10^{-8}$ below those obtained when $N = 4$. Due to this small effect, we used $N=4$ throughout the paper.
\end{remark}

\begin{remark}
All the time domain filters and covariance functions used in this paper were numerically calculated using the adaptive Lobatto quadrature method. For computational simplicity, we used Daubechies filters, instead of Meyer. In Matlab, the quadrature method is implemented via the \texttt{quadl.m} function. Section \ref{s:accuracy_numerical_integration} provides a numerical study of the accuracy of the \texttt{quadl.m} function comparing the deviation of the numerically computed AR(1) and FARIMA filters from their closed form expressions. We also performed experiments with the adaptive Gauss-Kronrod quadrature, which is more suitable for functions with moderate singularities at endpoints. This method is implemented in Matlab via the \texttt{quadgk.m} function and is also studied in Section \ref{s:accuracy_numerical_integration}.  Note that, alternatively, the computation of Fourier transforms of functions displaying a singularity at the origin can be dealt with via a change of variables. See Helgason et al.\ \cite{helgason:pipiras:abry:2011}, Section 3.4.
\end{remark}

%

%

\section{Evaluating the simulation method}\label{s:eval_sim}

In this section, we evaluate the accuracy of the simulation method. Most simulation techniques are supported by theorems that establish some sort of convergence, equality in law and so on. However, the finite sample performance can be disparate across methods in practice. One approach is to use estimators  to compare the simulation methods. Nevertheless, only the asymptotic distribution of estimators are available in most cases. The finite sample performance of estimators, e.g., bias, is then studied based on simulation, which creates a circularity.  In view of this, we study the performance of the methods relative to one another and compare with three other methods: simple iteration (OU process), Cholesky and CME.

Cholesky decompositions provide a classical and simple simulation method. If $V$ is a target zero mean Gaussian stationary process with a given, known covariance matrix $\Sigma$ over a finite set of time points $N$, a Cholesky decomposition of $\Sigma = LL^{*} $ is performed, where $L$ is a lower triangular matrix, and  vector $Z$ of i.i.d.\ standard Gaussian variables is generated.  Then $V \stackrel{d}= LZ$, as desired. Cholesky-based simulation is exact up to the accurate calculation of the covariance and can be implemented recursively (see Asmussen and Glynn \cite{asmussen:glynn:2000}; see also Bardet et al.\ \cite{bardet:lang:oppenheim:philippe:taqqu:2003}, Craigmile \cite{craigmile:2005}). However, it is slow in terms of computational complexity: $O(N^3)$, where $N$ denote the length of the resulting stochastic vector.

Another popular method is the CME (see Davies and Harte \cite{davies:harte:1987}, Wood and Chan \cite{wood:chan:1994}, Dietrich and Newsam \cite{dietrich:newsam:1997}, Johnson \cite{johnson:1994}, Beran \cite{beran:1994}, Asmussen and Glynn \cite{asmussen:glynn:2000}, Percival and Constantine \cite{percival:constantine:2002}, Craigmile \cite{craigmile:2003}). The algorithm involves embedding the covariance matrix in a non-negative definite circulant matrix of size $M \geq 2(N-1)$. This is computationally convenient, since the diagonalization of circulant matrices can be carried out by means of the Fast Fourier Transform (FFT), which has complexity $O(N \log(N))$. Like Cholesky-based simulation, CME is exact. For a description of the CME, see Bardet et al.\ \cite{bardet:lang:oppenheim:philippe:taqqu:2003}, p.\ 582.

 Since the OU process can be simulated based on a simple loop, we choose this method to provide the baseline for the CME and the wavelet-based method. In the cases of the fOU and fGLE process, the baseline method is Cholesky-based simulation, since it is also a simple and exact procedure. A two-sample $t$ statistic is used to assess the difference between the values of the estimator when generated by two of the methods. For the OU process, we evaluate the quality of the simulation based on the Yule-Walker estimator, whereas for the fOU and fGLE we use the Local Whittle estimation of the parameter $d$. The latter is of special interest in the framework of subdiffusion, since the Local Whittle is a good estimator for the subdiffusivity parameter $\alpha$ (see Didier et al.\ \cite{didier:mckinley:hill:fricks:2012}).

%


For the OU process, the initial, exact step $j=0$ amounts to simulating through a simple loop an AR(1) process with parameter $\phi = e^{- \zeta}$ and white noise variance $\frac{1 - e^{-2 \zeta}}{2 \zeta}$.  Based on both wavelet and exact simulation, we generated the OU process with parameters $\zeta = 1 $ and $\sigma = 1$ over the interval $[0,2^{8}]$, with $2^{13}$ points in each subinterval of length 1. Then, by sampling at the rate $\Delta = 2^{-3}$ (i.e., every $2^{10}$ points), the associated AR(1) process has parameter $\phi = \exp(- 1 \cdot 2^{-3}) = 0.8825$, estimated by Yule-Walker over a time series of total length $2^{11}$. In order to speed up the computations while preserving accuracy, the wavelet filters were truncated either at lag $|T| = 40$ or when a value below $10^{-9}$ is attained, whichever is first. In order to test the consistency of wavelet simulation for different values for the parameter $\zeta$ and thus different filters, we also generated the OU process with parameters $\zeta = 2 $ and $\sigma = 1$ over the interval $[0,2^{7}]$ and sampled it at $\Delta = 2^{-4}$ and then with parameters $\zeta = 1/2 $ and $\sigma = 1$ over the interval $[0,2^{9}]$ and sampled it at $\Delta = 2^{-2}$. Therefore, the associated AR(1) processes have the same parameter $\phi = \exp(- 2 \cdot 2^{-4})=\exp(- 1/2 \cdot 2^{-2}) = 0.8825$. The simulation results, found in Table \ref{t:sim_OU_08825}, suggest that the method is accurate when compared to iterative simulation and CME. The filters for $\zeta = 1/2$ seem to be less accurate than those for $\zeta = 1, 2$. When the final scale is $J = 4$, the absolute value of the $t$ statistic is above 4. However, when $J$ is increased to 6 and 8, the latter drops below 2. This is indicative of increasing quality of the discretization as an increasing function of $J$. In further, unshown computational work, we obtained similar results for $\zeta = 1$ when $\Delta = 2^{-2}$ and $\Delta = 2^{-4}$ (thus, $\phi = \exp(-2^{-2})$ and $\phi = \exp(-2^{-4})$, respectively), and also for time series of  total length $2^{9}$ instead of $2^{11}$




For the fOU process the comparison is made over the integer time points $0,1,2, \hdots$. For all simulations, filters were truncated at lag $T = 40$ giving entries on the order $10^{-6}$ at the point of truncation in the worst cases, typically in the low pass filters $u_j$. We experimented with two different initializations: either via the convolution of the filter $g_0 = v_0$ (i.e., at $j = 0$) with white noise or directly via CME with the associated autocovariance function calculated by means of numeric integration. The results can be seen in Table \ref{t:sim_fracOU_d_length9_initj0procCME} for the parameter values $d = 0.10, 0.25, 0.45$. In practice, the disadvantage of initialization via direct convolution with white noise is that in principle it might require storage of fairly long filters when $d > 0$ (i.e., under long range dependence). For this reason,  $g_0$ is truncated for some cases at $T = 1,200$ in Table \ref{t:sim_fracOU_d_length9_initj0procCME} yielding a quite large overall length of the filter $g_0$ at 2,401. The absolute value of the $t$ statistics is less than 2 regardless of the initialization (convolution or CME), thus yielding results rather similar to CME for the simulation of the fOU. Though not displayed in the tables, simulation initialized with truncated filters at $T = 400$ or $600$ seem to give rather similar results.

For $d < 0$, the discrete time fOU process is not anti-persistent (see Corollary \ref{c:increm_subdiff_specdens}). To evaluate the wavelet-based simulation procedure, we took the discrete time increment $Y(n) = X(n) - X(n-1)$ of the associated position process $X(t) = \int^{t}_{0}V(s)ds$. The spectral density of $Y(n)$ is, indeed, anti-persistent for $d < 0$ (see Corollary \ref{c:increm_subdiff_specdens}). For wavelet simulation purposes, the process $X(\cdot)$ was approximated by the simulated $V(\cdot)$ based on the expression (\ref{e:sup|X - sum V|}): for some fixed $J$, a sequence $V_{J,k}$ was generated, and the Riemann sum sequence in \eqref{e:sup|X - sum V|}
was then calculated and sampled. The Cholesky sequences were simulated based on the covariance function of the process $\Delta X(t)$ with $\sigma =1$,
$$
E \Delta X(s)\Delta X(s+t) = \frac{\Gamma(2d+2) \sin(\pi (d+1/2))}{2 \pi}  \int_{\bbR} e^{itx} \Big| \frac{1 - e^{-ix}}{ix} \Big|^2 \frac{1}{\zeta^2 + x^2} \frac{1}{|x|^{2d}} dx.
$$
The results are also shown in Table \ref{t:sim_fracOU_d_length9_initj0procCME} for different parameter values $d = -0.10, -0.25, -0.45$. Once again, similarly to CME, the absolute value of the $t$ statistics is less than 2 in all cases.

Table \ref{t:sim_fracOU_length9_initj0procCME_compareJ} displays a study of the accuracy of the wavelet-based simulation of fOU as a function of the finest scale $J$ when $d = 0.25, -0.25$. Theoretically, as $J \rightarrow \infty $, the quality of the discretization improves. However, the results show that the relative bias does not change much as a function of $J$ for $J = 2,4,6,8,10$. This potentially indicates that the quality of the simulation is already good enough at low values of $J$.

\begin{figure}
	\begin{center}
	\includegraphics[height=2in,width=2.5in]{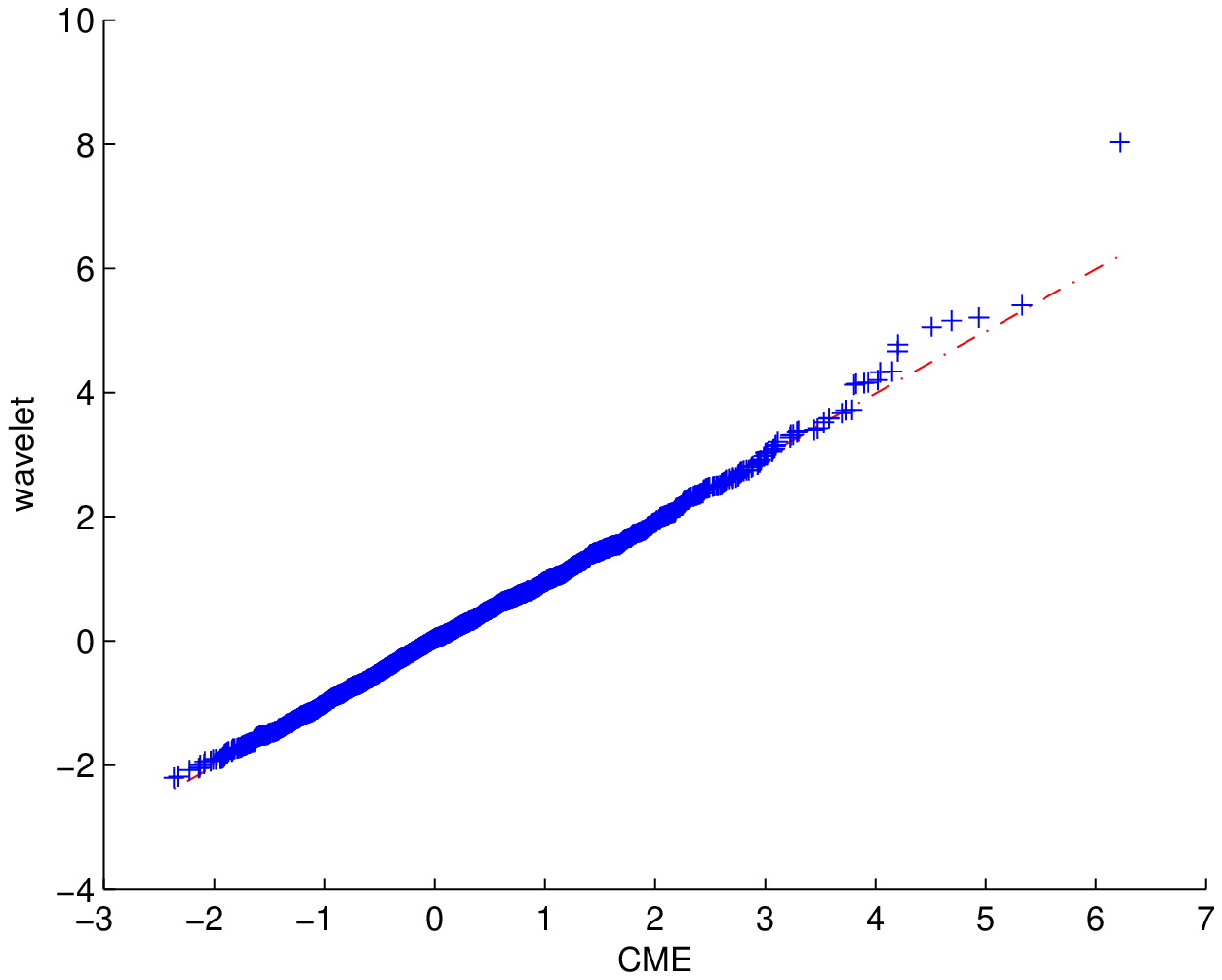} \ \includegraphics[height=2in,width=2.5in]{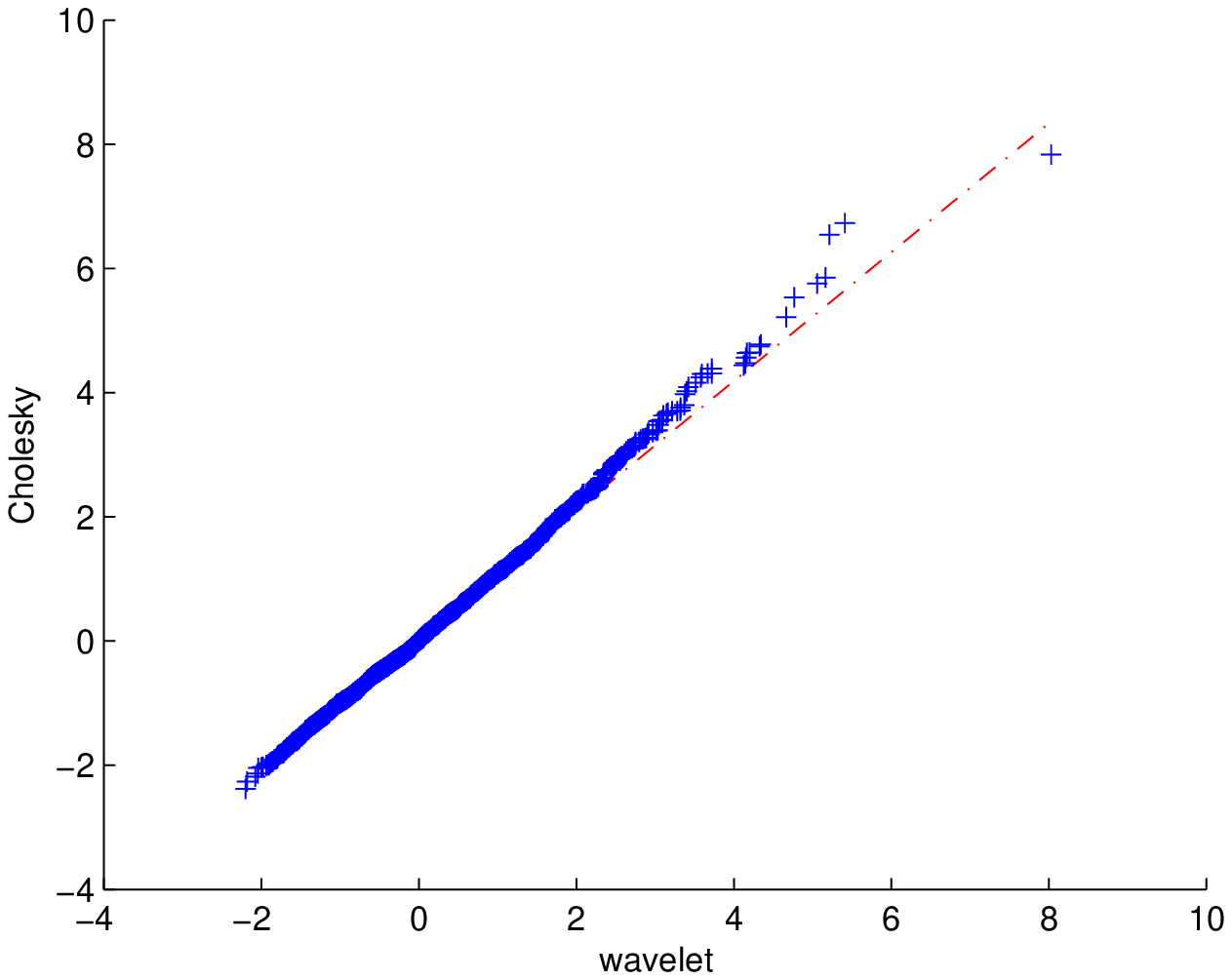}
	\caption{\label{f:qqplots} fOU ($\zeta = 1$, $d = 0.25$, time series length $2^9$), qq-plots for 3000 Monte Carlo runs of the statistic $\widetilde{T}$. Left plot: CME versus wavelet. Right plot: wavelet versus Cholesky.
	}
	\end{center}
\end{figure}

Since the fOU has long range dependence for $d=0.25$, which creates a more difficult case for simulation, we chose to also provide a non-parametric test for the full correlation structure of the process
$$
H_0: f(x) = f_V(x|\zeta = 1, d =0.25) \quad \textnormal{versus} \quad H_A: f(x) \neq f_V(x|\zeta = 1, d =0.25) ,
$$
where $f(x)$ is the actual spectral density and $f_V(x)$ is as in \eqref{e:sampled_specdens_fOU}. The test is given by the spectral statistic $\widetilde{T}$ as defined in Bardet et al.\ \cite{bardet:lang:oppenheim:philippe:taqqu:2003}, p.\ 595, and Chen and Deo \cite{chen:deo:2004}. This statistic satisfies the asymptotic relation $\widetilde{T} \stackrel{d}\rightarrow N(0,1)$ as the length of the time series goes to infinity. For each simulation method CME, Cholesky and wavelets, we did 3000 Monte Carlo runs of $\widetilde{T}$ via $2^9$ observations of the fOU process. The histogram of the distribution generated by the wavelet method can be seen in Figure \ref{f:spec_density_truncpatch}: at the chosen length $2^9$, the distribution is already mound-shaped, but still slightly skewed to the right. Though not shown, the histograms for CME and Cholesky look similar, with the latter displaying values somewhat more concentrated around the mode. For the sake of comparison, we show the qq-plots for CME versus wavelet method, and wavelet method versus Cholesky in Figure \ref{f:qqplots} (though not displayed, the qq-plot for CME versus Cholesky looks similar to these). The qq-plots do not indicate a substantial discrepancy in the generated distributions of $\widetilde{T}$. This seems to be confirmed by the $p$-values as given by the two-sample Kolmogorov-Smirnov statistic: for CME versus wavelet method, wavelet method versus Cholesky, and CME versus Cholesky, they were 0.712, 0.402 and 0.498, respectively.

For the fGLE, the filters displayed slower decay than those for the fOU process. For this reason, truncation was performed at lag $|T| = 80$, which gave entries of the order $10^{-6}$ at the point of truncation in most cases, typically in the low pass filters, and entries of $10^{-5}$ only for the low pass filter $u_1$ for different values of $d$. For the same reason as for the fOU, the wavelet simulation was performed based on the Riemann sum in \eqref{e:sup|X - sum V|} (see Corollary \ref{c:increm_subdiff_specdens}). Also, the Cholesky sequences were simulated based on the covariance function of the process $\Delta X(t)$,
$$
E \Delta X(s)\Delta X(s+t) = c  \int_{\bbR} e^{itx} \Big| \frac{1 - e^{-ix}}{ix} \Big|^2 \frac{1}{\gamma_0 + \gamma_1 |x|^{\beta} + \gamma_2 |x|^{2 \beta}} \frac{1}{|x|^{2d}} dx
$$
for an appropriate $c > 0$. The results are shown for different parameter values and final scales $J$ in Table \ref{t:sim_fracGLE_length9_initj0procCME_compareJ}. In all cases, the absolute value of the $t$ statistic is less than two and close to the corresponding value obtained from CME simulation. Also, in contrast with the OU process, increasing $J$ does not seem to affect considerably the quality of the simulation.

\begin{remark}
Other simulation studies were carried out for the fOU, $0 < d < 1/2$, with filters calculated via the adaptive Gauss-Kronrod quadrature. The results were comparable to those obtained via adaptive Lobatto quadrature, so they are not shown.
For all numerical integrals for all methods, the quadrature precision ranged from $10^{-9}$ to $10^{-13}$. Whenever applicable, the length of the initial CME-simulated series for the wavelet method was $2^{10}$.
\end{remark}

%

\section{Discussion}\label{s:discussion}

Wavelet-based simulation methods have proven to be fast and efficient alternatives to FFT-based methods for rather large samples. They usually exhibit low computational complexity, since they are based on the Fast Wavelet Transform (see Percival and Walden \cite{percival:walden:2000}). In this paper, we proposed an approximate wavelet-based simulation technique for two classes of continuous time anomalous diffusion models, the fGLE and the fOU. The proposed algorithm is an iterative method that provides approximate discretizations that converge quickly to the true sample path of the target process. The simulation technique involves an appropriate sequence of filters $g_{j}$, $j=0,1,...,J$, where is $J$ the finest scale chosen. The method then amounts to recursively generating an induced sequence of (exact or approximate) discretizations $V_j$, at scale $j \in \bbN$, where $V_j = g_j \ast \xi$. One can then show that $2^{J/2}V_{J,\lfloor 2^{J}t \rfloor } \rightarrow V(t)$, $J \rightarrow \infty$, in an appropriate sense, which naturally leads to an approximation of the position process $X(t)$ via Riemann sums. As compared to previous works such as Didier and Pipiras \cite{didier:pipiras:2008}, this paper proposes a simulation procedure when the discretization of the target continuous time process at different scales does not have closed form in the time domain. Moreover, we propose smoothing procedures for the proposed filters as to speed their time domain decay, and thus minimize the border effect, which is always present in convolution-based procedures.

While this method is approximate, it has several advantages such as: computational speed; discretizations that converge uniformly over compact intervals almost surely; it is iterative; it is not intrinsically Gaussian. To study the performance of the wavelet-based simulation in comparison to exact methods such as Cholesky and CME, we performed several Monte Carlo experiments. The simulation study measured the bias of well-established estimators when compared with realizations from other methods. In most cases, the bias of the estimators when simulated using the wavelet method seems to lie within an insignificant distance from the exact methods.  Therefore, the method is nearly as accurate with potentially reduced computational complexity compared with existing methods.

\appendix

\section{Proofs}\label{s:adaptation_proofs}

In this section, we discuss the adaptation of the original proofs in Didier and Pipiras \cite{didier:pipiras:2008}.

We will need to replace Assumption 3 with the weaker Assumption $3'$, which allows for a kink in $G_{j}$ and $G^{-1}_{j}$ at $\pm \pi$,
Assumption $4'$ replaces Assumption 4 for the sake of clarification. We remind the reader that the support of the Meyer scaling function $\widehat{\phi}(x)$ is contained in the interval $[-4\pi/3,4\pi/3]$.

\medskip


\noindent {\sc Assumption $3'$:}
$$
G_j(x),G_j(x)^{-1} \in C[-4\pi/3,4 \pi/3] \cap C^{2}([-4\pi/3,4\pi/3] \backslash \{\pm \pi\})
$$
\begin{equation}\label{e:A3'_ii}
\max_{p=-1,1} \sup_{j \geq j_0} \sup_{|x|\leq 4\pi/3}|G_{j}(x)^{p}| < \infty
\end{equation}
\begin{equation}\label{e:A3'_iii}
\max_{p=-1,1} \max_{k=1,2}\sup_{j \geq j_0} \sup_{|x|\neq \pi, |x|\leq 4\pi/3} \Big|\frac{\partial^{k}}{\partial x^k }[G_{j}(x)^{p}] \Big| < \infty
\end{equation}
\begin{equation}\label{e:A3'_iv}
\max_{p=-1,1} \max_{k=1,2}\sup_{j \geq j_0} \lim_{x \rightarrow \pi^+ ,-\pi^{-}} \Big|\frac{\partial^{k}}{\partial x^k }[G_{j}(x)^{p}] \Big| < \infty
\end{equation}

\medskip
\noindent {\sc Assumption $4'$:} $\widehat{g}$ is twice differentiable in $\bbR \backslash\{0\}$ and, for large $|x|$,
$$
\left| \frac{\partial^k \widehat g(x)}{\partial x^k}\right| \leq
\frac{\mbox{const}}{|x|^{k+1}}, \quad k=0,1,2.
$$
Assumption $4'$ is clearly satisfied by the spectral filters $\widehat{g}(x)$ of all the processes considered in this paper, so we turn to Assumption $3'$. Note that in \eqref{e:A3'_iv} the limits are assumed to exist.

In Lemma \ref{l:lemma1_in_AWD1}, we show that the conclusion of Lemma 1 in Didier and Pipiras \cite{didier:pipiras:2008} still holds if Assumption 3 is replaced with Assumption $3'$. For the reader's convenience, we discuss the modification of the proof. Before that, we establish a simple auxiliary lemma.

\begin{lemma}\label{l:integ_by_parts}
Let $f$,$g$ be two real functions in $C^1(a,b)$ such that $f'$, $g'$ are bounded over $[a,b]$ and the limits $\lim_{x \rightarrow a^{+},b^{-}}f(x)g(x)$ exist. Then the formula for integration by parts holds, i.e.,
$$
\lim_{\varepsilon_b \rightarrow 0^{+}}f(b-\varepsilon_b)g(b-\varepsilon_b) - \lim_{\varepsilon_a \rightarrow 0^{+}}f(a+
\varepsilon_a)g(a+\varepsilon_a) = \int^{b}_{a}[f'(x)g(x) + f(x)g'(x)]dx.
$$
\end{lemma}
\begin{proof}
Since $f, g \in C^1(a,b)$, then for small $\varepsilon_a, \varepsilon_b> 0$,
$$
f(b-\varepsilon_b)g(b-\varepsilon_b) - f(a+\varepsilon_a)g(a+\varepsilon_a) = \int^{b-\varepsilon_b}_{a + \varepsilon_a}[f'(x)g(x) + f(x)g'(x)]dx.
$$
By the remaining assumptions and taking the limits $\lim_{\varepsilon_a \rightarrow 0^+}$, $\lim_{\varepsilon_b \rightarrow 0^+}$, the claim follows. 
\end{proof}

\begin{lemma}\label{l:lemma1_in_AWD1}
Under Assumptions $3'$ and $4'$,
\begin{equation}\label{e:decay_Phi_j_Phi^j}
|2^{-j/2} \Phi_{j}(2^{-j}u)|, |2^{-j/2} \Phi^{j}(2^{-j}u)| \leq \frac{C}{1 + |u|^2}, \quad u \in \bbR
\end{equation}
\begin{equation}\label{e:decay_Psi^j}
|\Phi^{j}(2^{-j}u)| \leq \frac{C 2^{-j/2}}{1 + |u|^2}, \quad u \in \bbR,
\end{equation}
where $\Phi_{j}$, $\Phi^{j}$ and $\Psi^{j}$ are defined in the Fourier domain as $\widehat{\Phi}_{j}(x) = \overline{G_{j}(2^{-j}x)}2^{-j/2}\widehat{\phi}(2^{-j}x)$, $\widehat{\Phi}^{j}(x) = G_{j}(2^{-j}x)^{-1}2^{-j/2}\widehat{\phi}(2^{-j}x)$, $\widehat{\Psi}^j(x)= \widehat{g}(x) 2^{-j/2}\widehat{\psi}(2^{-j}x)$.
\end{lemma}
\begin{proof}
We first look at $\Phi^{j}$. Consider initially $x > 0$. From the properties of the Meyer MRA, $G_{j}(x)^{-1} \widehat{\phi}(x)$ also satisfies Assumption $3'$. Thus, by Lemma \ref{l:integ_by_parts},
$$
\int^{4\pi/3}_{\pi} e^{iux} G_{j}(x)^{-1} \widehat{\phi}(x) dx = \Big|^{4 \pi/3}_{\pi} \frac{e^{iux}}{iu} G_{j}(x)^{-1} \widehat{\phi}(x) - \int^{4\pi/3}_{\pi} \frac{e^{iux}}{iu} \frac{\partial}{\partial x}(G_{j}(x)^{-1} \widehat{\phi}(x)) dx
$$
\begin{equation}\label{e:bound_Phi^j}
= \frac{-1}{iu} \Big(e^{iu \pi} G_{j}(\pi)^{-1} \widehat{\phi}(\pi) + \int^{4\pi/3}_{\pi} e^{iux} \frac{\partial}{\partial x}(G_{j}(x)^{-1} \widehat{\phi}(x)) dx \Big).
\end{equation}
In turn, again by Lemma \ref{l:integ_by_parts}, the integral on the right-hand side of \eqref{e:bound_Phi^j} is
$$
\Big|^{4 \pi/3}_{\pi} \frac{e^{iux}}{iu} \frac{\partial}{\partial x}(G_{j}(x)^{-1} \widehat{\phi}(x)) -  \int^{4\pi/3}_{\pi}
\frac{e^{iux}}{iu} \frac{\partial^2}{\partial x^2}(G_{j}(x)^{-1} \widehat{\phi}(x))dx
$$
\begin{equation}\label{e:bound_Phi^j_second_integ_parts}
= - \frac{e^{iu \pi}}{iu} \lim_{x \rightarrow \pi^+} \frac{\partial}{\partial x}(G_{j}(x)^{-1} \widehat{\phi}(x)) -  \int^{4\pi/3}_{\pi}
\frac{e^{iux}}{iu} \frac{\partial^2}{\partial x^2}(G_{j}(x)^{-1} \widehat{\phi}(x))dx.
\end{equation}
As for $x < 0$, an analogous expression holds for
\begin{equation}\label{e:bound_Phi^j_-4pi/3_-pi}
\int^{-\pi}_{-4\pi/3} e^{iux} G_{j}(x)^{-1} \widehat{\phi}(x) dx.
\end{equation}
On the other hand,
$$
\int^{\pi}_{-\pi} e^{iux} G_{j}(x)^{-1} \widehat{\phi}(x) dx
$$
\begin{equation}\label{e:bound_Phi^j_-pi_pi}
= \frac{1}{iu }\Big( e^{iu\pi}G_{j}(\pi)^{-1}\widehat{\phi}(\pi) - e^{-iu\pi}G_{j}(-\pi)^{-1}\widehat{\phi}(-\pi) -  \int^{\pi}_{-\pi}
e^{iux} \frac{\partial}{\partial x} (G_{j}(x)^{-1}\widehat{\phi}(x))dx \Big).
\end{equation}
Once again by Lemma \ref{l:integ_by_parts}, the integral on the right-hand side of \eqref{e:bound_Phi^j_-pi_pi} is
$$
\frac{1}{iu} \Big( e^{iu \pi} \lim_{x \rightarrow \pi^{+}}  \frac{\partial}{\partial x}(G_{j}(x)^{-1} \widehat{\phi}(x)) - \lim_{x \rightarrow -\pi^{-}} e^{-iu \pi} \frac{\partial}{\partial x}(G_{j}(x)^{-1} \widehat{\phi}(x)\Big)
$$
\begin{equation}\label{e:bound_Phi^j_-pi_pi_second_integ_parts}
- \int^{\pi}_{-\pi} e^{iux} \frac{\partial^2}{\partial x^2}(G_{j}(x)^{-1} \widehat{\phi}(x)) dx \Big).
\end{equation}
Consequently, by adding together \eqref{e:bound_Phi^j}, \eqref{e:bound_Phi^j_second_integ_parts}, \eqref{e:bound_Phi^j_-4pi/3_-pi}, \eqref{e:bound_Phi^j_-pi_pi}, \eqref{e:bound_Phi^j_-pi_pi_second_integ_parts} and by Assumption $3'$,
$$
|2^{-j/2} \Phi^{j}(2^{-j}u)| = \frac{1}{2\pi} \Big| \Big( \int^{4\pi/3}_{\pi} + \int^{\pi}_{-\pi}  + \int^{-\pi}_{-4\pi/3}\Big)  e^{iux} G_{j}(x)^{-1} \widehat{\phi}(x) dx \Big| \leq \frac{C}{2 \pi u^2 }
$$
for a constant $C$ that does depend on $j$, which gives the inequalities in \eqref{e:decay_Phi_j_Phi^j} for $\Phi^j$. The remaining inequality, for $\Phi_j$, can be obtained by a similar procedure.

To show \eqref{e:decay_Psi^j}, start from
$$
\Psi^j (2^{-j}u) = \frac{2^{j/2}}{2 \pi} \int_{\bbR}e^{iux}\widehat{g}(2^{j}x)\widehat{\psi}(x) dx, \quad u \in \bbR.
$$
Since the only possible singularity of $\widehat{g}$ is at the origin by Assumption $4'$ and $\textnormal{supp}\{\widehat{\psi}\}\subseteq \{2 \pi/3 \leq |x| \leq 8 \pi/3\}$, the same argument as in the proof of Lemma 1 in Didier and Pipiras \cite{didier:pipiras:2008} applies, thus yielding \eqref{e:decay_Psi^j}. 
\end{proof}

Proposition \ref{p:decay_AWD_filters} is a claim in the proof of Proposition 1 in Didier and Pipiras \cite{didier:pipiras:2008}. Since we use filters under slightly different assumptions in this paper, we state it and provide a more detailed proof.\\

{\sc Proof of Proposition \ref{p:decay_AWD_filters}}:
The Meyer low-pass filter $u$ satisfies $\textnormal{supp}\{\widehat{u}\} \subseteq \{|x| \leq 2\pi/3\}$.
On the other hand,  the possible kinks of $G_{j+1}(x)$ in $[-4\pi/3,4\pi/3]$ lie at $\pm \pi$ by Assumption $3'$. Therefore, the potential kinks of $G_{j}(2x)$ in $[-2\pi/3,2\pi/3]$ lie at $\pm \pi/2$. However, those points lie outside $\textnormal{supp}\{\widehat{u}\}$. Therefore, we can write
$$
u_{j,k} = \Big(\int^{-\pi/2}_{-2\pi/3} + \int^{\pi/2}_{-\pi/2} + \int^{2\pi/3}_{\pi/2}\Big) e^{ikx} \frac{G_{j+1}(x)}{G_{j}(2x)} \widehat{u}(x) dx,
$$
and again by Assumption $3'$, one can use the same type of argument as in the proof of Lemma \ref{l:lemma1_in_AWD1} to establish \eqref{e:u,v_decay} for $u_{j,k}$.

As for $v_{j,k}$, $\widehat{g}_{j+1}(x)$ is smooth except possibly at the origin by Assumption $4'$. Since $\widehat{v}_{j}(x) = \widehat{g}_{j+1}(x) \widehat{v}(x)$, $\textnormal{supp}\{\widehat{v}\}\subseteq \{\pi/3 \leq |x| \leq 5 \pi/3\}$ and the fact that $\widehat{v}(x) \in C^{2}[-\pi,\pi)$, then by applying integration by parts twice we arrive at the claim. $\Box$\\

We now describe the necessary modifications to the remaining claims in Didier and Pipiras \cite{didier:pipiras:2008}:
\begin{itemize}
\item Theorem 2, section 5, expression (5.13) we still have that
$$
\widehat{F}_{m}(x) = \Big( \sum^{m}_{k=-m}g_{J,k}e^{-i 2^{-J}kx}\Big)\frac{\widehat{g}(x)}{\widehat{g}_{J}(2^{-J}x)}2^{-J/2}\widehat{\phi}(2^{-J}x) \rightarrow \widehat{\theta}^{J}(x), \quad m \rightarrow \infty
$$
in $L^2(\bbR)$, since $\sum^{m}_{k=-m}g_{J,k}e^{-i k x}$ converges to $\widehat{g}_{J}(x)$ in $L^{2}[-\pi,\pi)$ and, by Assumption $3'$, $\frac{\widehat{g}(x)}{\widehat{g}_{J}(2^{-J}x)}$ is bounded on the compact support of $\widehat{\phi}(2^{-J}x)$.
\item Proposition 1, section 6, expression (6.5): as an immediate consequence of Lemma \ref{p:decay_AWD_filters},
$$
\widehat{u}_{j}(x) = G_{j+1}(x) (G_{j}(2x))^{-1} \widehat{u}(x) \in L^2[-\pi,\pi),
$$
$$
\widehat{v}_{j}(x) = G_{j+1}(x) \widehat{g}(2^{j+1}x) \widehat{v}(x) \in L^2[-\pi,\pi).
$$
By a similar proof to that for Lemma \ref{p:decay_AWD_filters}, we also conclude that
$$
\widehat{u}^{d}_{j}(x) = \overline{G_{j+1}(x)^{-1}} \overline{G_{j}(2x)} \widehat{u}(x) \in L^2[-\pi,\pi),
$$
$$
\widehat{v}^{d}_{j}(x) = \overline{G_{j+1}(x)^{-1}}  \overline{\widehat{g}(2^{j+1}x)^{-1}} \widehat{v}(x) \in L^2[-\pi,\pi).
$$
\end{itemize}

We now show that the proposed filters satisfy Assumption $3'$. Let
\begin{equation}\label{e:f_j(x)}
f_{j}(x) = G_{j,\gamma,d}(x)^{2}.
\end{equation}
Then $f_{j}(x)$ satisfies \eqref{e:A3'_ii} (i.e., $f_{j}$ and its inverse is uniformly bounded over $|x| \leq 4\pi/3$ and large $j$). Moreover, since
$$
G^{'}_{j,\gamma,d}(x) = \frac{1}{2}f_{j}(x)^{-1/2}f^{'}_{j}(x), \quad  G^{''}_{j, \gamma,d}(x) = \frac{1}{2}\Big( -\frac{1}{2}f_{j}(x)^{-3/2}f^{'}_{j}(x)^{2} + f_{j}(x)^{-1/2}f^{''}_{j}(x) \Big),
$$
then it suffices to look at $f^{'}_{j}(x)$ and $f^{''}_{j}(x)$.

\begin{lemma}\label{l:Gj_satisfies_A2_A3'_A5}
Let $G_{j,\zeta}(x)$, $G_{j,d}(x)$ and $G_{j,\gamma,d}(x)$ be as in \eqref{e:G_{j,a}}, \eqref{e:G_{j,d}} and \eqref{e:G_{j,gamma,d}}, respectively. Then all these filters satisfy Assumptions $2$, $3'$, and $5$.
\end{lemma}
\begin{proof}
It is clear that Assumptions 2 and 5 are satisfied in all cases, so we focus on Assumption $3'$.

The argument is straightforward for $G_{j,\zeta}(x)$ and $G_{j,d}(x)$. As for $G_{j,\gamma,d}(x)$, let $f_j(x)$ be as in \eqref{e:f_j(x)}. Then it suffices to show that it satisfies \eqref{e:A3'_iii} and \eqref{e:A3'_iv} (with $f_{\cdot}$ in place of $G_{\cdot}$). Note that $f_{j}(x) = 1$, $|x| \leq \pi$, and thus the first and second derivatives are trivial in this range. Without loss of generality, we now only look at the range $\pi \leq x \leq 4\pi/3$, where $x$ goes to $\pi^{+}$ as a side limit. In this case,
$$
f_{j}(x) = \frac{\gamma_0 + \gamma_1 2^{j\beta} x^{\beta} + \gamma_2 2^{2j\beta}x^{2 \beta}}{\gamma_0 + \gamma_1 2^{j \beta }(2\pi - x)^{\beta} + \gamma_2 2^{2 j \beta}(2\pi - x)^{2 \beta}}.
$$
Thus,
$$
f^{'}_{j}(x) = (\gamma_1 2^{j \beta} \beta x^{\beta-1} + \gamma_2 2^{2j\beta}2\beta x^{2\beta-1})
(\gamma_0 + \gamma_1 2^{j\beta}(2\pi - x)^{\beta} + \gamma_2 2^{2j\beta}(2\pi-x)^{2 \beta})^{-1}
$$
$$
+ (-1)^2 (\gamma_0 + \gamma_1 2^{j\beta} x^{\beta} + \gamma_2 2^{2j\beta}x^{2 \beta})
 (\gamma_0 + \gamma_1 2^{j \beta}(2\pi -x)^{\beta} + \gamma_2 2^{2j\beta}(2\pi - x)^{2\beta})^{-2}
 $$
 $$
 \cdot (\gamma_1 2^{j \beta}\beta (2\pi -x)^{\beta-1} + \gamma_2 2^{2j\beta}2\beta (2\pi-x)^{2 \beta-1})
 $$
\begin{equation}\label{e:f'j(x)}
=: a_{j}(x)b_j(x) + c_j(x)d_{j}(x)e_j(x).
\end{equation}
Note that $\lim_{x \rightarrow \pi^{+}}$ exists, a requirement in \eqref{e:A3'_iv}. As for the first term in the sum \eqref{e:f'j(x)}, note that, for any $\varepsilon > 0$, for large $j$, $|w_j |= |\gamma_0/ 2^{2 j \beta} + \gamma_1/2^{j \beta}(2\pi -x)^{\beta}| < \varepsilon$ uniformly in $x$ over $\pi \leq x \leq 4 \pi/3$. Moreover, since $\gamma_2 = m^2 > 0$, then $\gamma_2 (2\pi -x)^{2 \beta}$ attains its (constrained) minimum at $x = 4\pi/3$. Therefore,
$$
|b_{j}(x)| \leq |\gamma_2 (2 \pi - 4\pi/3)^{2 \beta} - |w_j||^{-1}2^{-2 j \beta},
$$
whereas
\begin{equation}\label{e:f'j(x)_1st_term_bound}
|a_{j}(x)| \leq 2^{2j\beta} \Big( \frac{|\gamma_1|}{2^{j \beta}} \beta \Big(\frac{4\pi}{3}\Big)^{\beta-1} + |\gamma_2| 2 \beta \Big(\frac{4\pi}{3}\Big)^{2 \beta-1} \Big).
\end{equation}
Now consider the absolute value of the second term in the sum \eqref{e:f'j(x)}. By a similar reasoning,
$$
|c_j(x) d_j(x) e_j(x)|\leq 2^{2j \beta} \Big( \frac{\gamma_0}{2^{2j\beta}} + \frac{|\gamma_1|}{2^{j\beta}} \Big( \frac{4 \pi}{3}\Big)^{\beta} + \gamma_2
\Big( \frac{4 \pi}{3}\Big)^{\beta}\Big)
2^{(-2)2j\beta} |\gamma_2(2\pi - 4\pi/3)^{2 \beta} - |w_j||^{-2}
$$
\begin{equation}\label{e:f'j(x)_2st_term_bound}
\cdot 2^{2j\beta}\Big( \frac{|\gamma_1|}{2^{j \beta}}\beta (2\pi - 4\pi/3)^{\beta-1} + \gamma_2 2 \beta (2 \pi - 4\pi/3)^{2 \beta-1}\Big).
\end{equation}
By \eqref{e:f'j(x)_1st_term_bound} and \eqref{e:f'j(x)_2st_term_bound}, $| f^{'}_{j}(x) |$ (and thus also $(f_{j}(x)^{-1})^{'}$) is uniformly bounded over $\pi \leq x \leq 4 \pi/3$ and large $j$, i.e., it satisfies \eqref{e:A3'_iii} and \eqref{e:A3'_iv} for $k=1$.

We now turn to the second derivative. We obtain
$$
f^{''}_{j}(x) = \Big[ (\gamma_1 2^{j \beta} \beta (\beta -1) x^{\beta-2} + \gamma_2 2^{2j \beta}2 \beta(2\beta -1) x^{2\beta-2})
b_j(x)
+ a_j(x) (-1)^2 (\gamma_0 + \gamma_1 2^{j \beta}(2 \pi -x)^{\beta}
$$
$$
 + \gamma_2 2^{2j \beta}(2\pi -x)^{2\beta})^{-2}
\cdot (\gamma_1 2^{j \beta}\beta(2\pi - x)^{\beta-1}
+ \gamma_2 2^{2j\beta}2\beta(2\pi-x)^{2\beta-1})\Big]
$$
$$
+
$$
$$
(\gamma_1 2^{j \beta}\beta x^{\beta-1} + \gamma_2 2^{2j\beta} 2\beta x^{2\beta-1})[d_j(x)e_j(x)] + c_j(x)[d_j(x)e_j(x)]^{'},
$$
where
$$
[d_j(x)e_{j}(x)]^{'} = (-2) (\gamma_0 + \gamma_1 2^{j \beta}(2\pi-x)^{\beta} + \gamma_2 2^{2j \beta}(2\pi-x)^{2 \beta})^{-3}
$$
$$
(\gamma_1 2^{j \beta}\beta (2 \pi - x)^{\beta-1}(-1) + \gamma_2 2^{2j\beta}2\beta(2\pi-x)^{2 \beta-1}(-1)) e_j(x)
$$
$$
+ d_j(x) (\gamma_1 2^{j\beta} \beta (\beta-1)(2 \pi -x)^{\beta-2}(-1) + \gamma_2 2^{2 j \beta}2\beta (2 \beta-1) (2 \pi -x)^{2 \beta-2}(-1)).
$$
As with the first derivative, $\lim_{x \rightarrow \pi^{+}}$ exists. Moreover, by a similar reasoning to that for $f^{'}_{j}(x)$, we therefore conclude that $| f^{''}_{j}(x) |$ (and thus $(f_{j}(x)^{-1})^{''}$) is uniformly bounded over $\pi \leq x \leq 4 \pi/3$ and large $j$. Therefore, \eqref{e:A3'_iii} and \eqref{e:A3'_iv} are also satisfied for $k=2$. 
\end{proof}

We are now in position to prove Theorem \ref{t:VJ_conv_unif}.\\
{\sc Proof of Theorem \ref{t:VJ_conv_unif}}:
By the proof of Proposition 2, section 7, in Didier and Pipiras \cite{didier:pipiras:2008}, we have to show that $V$ satisfies Assumptions 2, 5, and a H\"{o}lder condition, and make use of Assumption $3'$ instead of Assumption $3$. In fact, Assumptions 2 and 5 are satisfied by Lemma \ref{l:Gj_satisfies_A2_A3'_A5}. In Lemma \ref{l:X=int_V} we establish the H\"{o}lder condition \eqref{e:Holder} for some $\nu \in (0,1)$ in the cases of the fGLE and fOU. Finally, we can bound $|2^{-J/2}\Phi_{J}(2^{-J}v)|$ based on Lemma \ref{l:lemma1_in_AWD1}, which is a consequence of Assumptions $3'$ and $4'$. The latter is satisfied for the spectral filters $\widehat{g}(x)$ of the processes in question, whereas the former also holds in view of Lemma \ref{l:Gj_satisfies_A2_A3'_A5}. Thus, the claim follows. $\Box$\\

\begin{remark}Note that, in the case of this paper, Assumptions 6, $3^*$, $5^*$, 7 of Didier and Pipiras \cite{didier:pipiras:2008} are not used since the H\"{o}lder continuity order given by Lemma \ref{l:X=int_V} is $\nu \in (0,1)$.
\end{remark}

To prove Corollary \ref{c:X_Riemann}, without loss of generality we assume that $T \in \bbN$.\\
{\sc Proof of Corollary \ref{c:X_Riemann}}
By Lemma \ref{l:X=int_V}, $X(t)$ is well-defined as the integral \eqref{e:X=Riemann_V}.
Fix $J > 0$ and form an associated partition $\Big\{\frac{k}{2^{J}T}\Big\}_{k = 0,\hdots,2^{J}T-1}$ of $[0,T]$. Assume that $T \geq t \geq 1$. The case where $t < 1$ can be handled similarly.

On the one hand,
$$
\Big| \sum^{\lfloor(2^J-1)t\rfloor}_{k=0} V\Big( \frac{k}{2^J}\Big)\frac{1}{2^J} - \sum^{\lfloor(2^J-1)t\rfloor}_{k=0} 2^{J/2}V_{J,k}\frac{1}{2^J} \Big|
\leq \sum^{\lfloor(2^J-1)t\rfloor}_{k=0} \Big| V\Big( \frac{k}{2^J}\Big) -  2^{J/2}V_{J,k} \Big| \frac{1}{2^J}
$$
\begin{equation}\label{e:conv_X_1st_term}
\leq A_1 2^{-J\nu} \frac{\lfloor(2^J-1)t\rfloor + 1}{2^J},
\end{equation}
where the last inequality follows from Proposition \ref{t:VJ_conv_unif}.

On the other hand, since $t > 1$, then $\frac{\lfloor(2^{J}-1)t \rfloor + 1}{2^J} \leq t$. Therefore,
$$
\Big| \int^{\frac{\lfloor(2^{J}-1)t \rfloor + 1}{2^J}}_{0}V(s)ds - \sum^{\lfloor(2^J-1)t\rfloor}_{k=0}V\Big(\frac{k}{2^J}\Big)\frac{1}{2^J} \Big|
$$
$$
= \Big| \int^{\frac{\lfloor(2^{J}-1)t \rfloor + 1}{2^J}}_{0}V(s)ds - \sum^{\lfloor(2^J-1)t\rfloor}_{k=0} \int^{(k+1)/2^J}_{k/2^J}V\Big(\frac{k}{2^J}\Big) ds \Big|
$$
$$
\leq \sum^{\lfloor(2^J-1)t\rfloor}_{k=0} \int^{(k+1)/2^J}_{k/2^J} \Big| V(s) - V\Big(\frac{k}{2^J}\Big) \Big| ds \leq \sum^{\lfloor(2^J-1)t\rfloor}_{k=0} \sup_{s \in \lfloor\frac{k}{2^J},\frac{k+1}{2^J}\rfloor} \Big| V(s) - V\Big(\frac{k}{2^J}\Big) \Big| \int^{(k+1)/2^J}_{k/2^J} ds
$$
$$
\leq \max_{k=0,1,\hdots,\lfloor(2^J-1)t\rfloor}\sup_{s \in \lfloor\frac{k}{2^J},\frac{k+1}{2^J}\rfloor} \Big| V(s) - V\Big(\frac{k}{2^J}\Big) \Big| \frac{\lfloor(2^J-1)t\rfloor+1}{2^J}
$$
\begin{equation}\label{e:conv_X_2nd_term}
\leq \max_{k=0,1,\hdots,\lfloor(2^J-1)t\rfloor}\sup_{s \in \lfloor\frac{k}{2^J},\frac{k+1}{2^J}\rfloor} A_2 \Big| s - \frac{k}{2^J} \Big|^{\nu} \frac{\lfloor(2^J-1)t\rfloor+1}{2^J} \leq A_3 2^{-j \nu}
\end{equation}
by the H\"{o}lder condition \eqref{e:Holder}, where $A_2$, $A_3$ are random variables that depend only on $T$. Also,
by the H\"{o}lder continuity of $V$ (Lemma \ref{l:X=int_V}),
\begin{equation}\label{e:convergence_proof_residual_integ_term}
\Big| \int^{t}_{\frac{\lfloor(2^{J}-1)t \rfloor + 1}{2^J}} V(s)ds \Big| \leq \sup_{s \in [0,T]}{|V(s)|} \Big( t - \frac{\lfloor(2^{J}-1)t \rfloor + 1}{2^J} \Big) \leq \sup_{s \in [0,T]}{|V(s)|} \hspace{1mm}C 2^{-J}
\end{equation}
for some constant $C>0$.

Now take $\sup_{t \in [0,T]}$ on both sides of \eqref{e:conv_X_1st_term}, \eqref{e:conv_X_2nd_term}, \eqref{e:convergence_proof_residual_integ_term}. The claim follows from the triangle inequality. $\Box$\\


\section{Auxiliary results}\label{s:aux}

In this section, we develop the Fourier domain integral representations for $X$. We first establish that $X$ can be regarded as the integral of $V$.

\begin{lemma}\label{l:X=int_V}
Let $\{V(t)\}_{t \geq 0}$ be the velocity process for the fGLE \eqref{p:V_gle_spec_repres} or the fOU \eqref{e:fracOU_spec}. Let
\begin{equation}\label{e:X=Riemann_V}
X(t) = \int^{t}_{0}V(s)ds,
\end{equation}
where the integral (\ref{e:X=Riemann_V}) is taken in the Lebesgue sense. Then (\ref{e:X=Riemann_V}) is well-defined in the sense that there exists a process $\eta(t)$ which is equivalent to $V(t)$, and which satisfies the H\"{o}lder condition
\begin{equation}\label{e:Holder}
|\eta(t) - \eta(s)| \leq A |t-s|^{\nu} \quad a.s.
\end{equation}
for some random variable $A$ that only depends on $[0,T]$ and some $\nu \in (0,1)$.

Let $\{X(t)\}_{t \geq 0}$ be the position process associated with $\{V(t)\}_{t \geq 0}$. Then
\begin{equation}\label{e:X_integ_repres}
X(t) \stackrel{{\mathcal L}}= \int_{\bbR} \Big(\frac{e^{itx}-1}{ix}\Big) \widehat{g}(x) \widetilde{B}(dx), \quad 0 < \delta
< \frac{1}{2},
\end{equation}
holds, where $\widehat{g}(x)$ is a spectral filter and $\widetilde{B}(dx)$ is given in \eqref{e:Brownian_measure}.
\end{lemma}
\begin{proof}
We only look at the fGLE, since the argument for the fOU can be developed along the same lines.

We first show \eqref{e:X=Riemann_V}. From Cram\'er and Leadbetter \cite{cramer:leadbetter:1967}, pp.\ 181-182, the conclusion follows from verifying the condition
$$
\int^{\infty}_{0}x^{2 \nu}\log(1+x)|\widehat{g}(x)|^2 dx <
\infty
$$
from some $\nu \in (0,1)$. Let $\varepsilon > 0$. Since $|\widehat{g}(.)|^2$ in \eqref{p:V_gle_spec_repres} is continuous, the only potentially
problematic points are the origin or $\infty$.
As $x \rightarrow 0^+$, $|\widehat{g}(x)|^2
\sim |x|^{2d}$, so $\int^{\varepsilon}_{0}x^{2 \nu} \log(1+x)|\widehat{g}(x)|^2dx < \infty$ for $\nu > 0$.
As $x \rightarrow \infty$, $|\widehat{g}(x)|^2 \sim x^{-2(\beta - d)}$, so $\int^{\infty}_{\varepsilon} x^{2 \nu}
\log(1+x)|\widehat{g}(x)|^2dx < \infty$ for $\nu < d + 1/2$.

To show \eqref{e:X_integ_repres}, note that $s,s' \geq 0$, $E|V(s)V(s')| \leq \sqrt{E(V(s))^2}\sqrt{E(V(s'))^2}$, which is finite and constant, by stationarity. Now apply Fubini's Theorem and formula \eqref{e:X=Riemann_V}. 
\end{proof}

The following is a corollary to Lemma \ref{l:X=int_V}.

\begin{corollary}\label{c:increm_subdiff_specdens}
Let $\{X(t)\}_{t \geq 0}$ be the position process associated with the fGLE \eqref{p:V_gle_spec_repres}. Denote the discrete-time first difference process by $Y_{n}=\Delta X(n) = X(n+1) - X(n)$, $n \in \bbN \cup \{0\}$. Thus, its spectral density $f_Y$ is
$$
f_{Y}(x) = c(d)^2 \hspace{1mm}
\Big( \Big|\frac{e^{ix}-1}{ix}\Big|^2 \Big|\frac{1}{\zeta
\kappa(x)|x|^{-2d}-imx}\Big|^2 |x|^{-2 d}
$$
$$ +
|e^{ix}-1|^2 \sum_{k \in \bbZ\backslash\{0\}} \Big|\frac{1}{\zeta
\kappa(x+2\pi k)|x+2\pi k|^{-2d}-im(x+2\pi k)}\Big|^2 |x+2\pi k
|^{-2(d+1)}\Big).
$$
Let $\{V(t)\}_{t \geq 0}$ and $\{X(t)\}_{t \geq 0}$ be the velocity and position processes, respectively, associated with the fOU \eqref{e:fracOU_spec}. Then the spectral density of the discrete time process $\{V(t)\}_{t = 0, 1, \hdots}$ is
$$
f_{V}(x) = \sigma^2 \Gamma(2d+2) \sin(\pi(d + 1/2))
$$
\begin{equation}\label{e:sampled_specdens_fOU}
\cdot \Big( \frac{1}{\zeta^2 + x^2} \frac{1}{|x|^{2d}} + \sum_{k \in \bbZ \backslash\{0\}} \frac{1}{\zeta^2 + (x+2k\pi)^2} \frac{1}{|x+2k\pi|^{2d}} \Big), \quad x \in [-\pi,\pi).
\end{equation}
Denote the discrete-time first difference process by $Y_{n}=\Delta X(n) = X(n+1) - X(n)$, $n \in \bbN \cup \{0\}$. Thus, its spectral density $f_Y$ is
$$
f_{Y}(x) = \sigma^2 \Gamma(2d+2) \sin(\pi (d+1/2))\Big( \Big| \frac{1 - e^{-ix}}{ix} \Big|^2 \frac{1}{\zeta^2 + x^2} \frac{1}{|x|^{2d}}
$$
$$
+ | 1 - e^{-ix}|^2  \sum_{k \in \bbZ \backslash\{0\} } \frac{1}{\zeta^2 + (x+2k\pi)^2} \frac{1}{|x+2k\pi|^{2d+2}}\Big),  \quad x \in [-\pi,\pi).
$$
\end{corollary}

\newpage

\section{Tables}\label{s:tables}

\begin{table}[h]
\caption{fOU: Local Whittle estimation of $d$ ($\zeta=1$) (wavelet filters cut off at lag $|T| = 40$, time series length $2^{9}$ for all methods)}
\centering
\begin{tabular}{ccccc}
\hline
$d = 0.10$ & $\widehat{d}$ & $s$ & $N$ & $|t|$ statistic \\ \hline
wavelet ($g_0$ length $2\times 400+1$, $J=6$) & 0.12106833 & 0.09259765 & 5000 & 0.01292907\\
wavelet (CME at  $j = 0$,  $J=6$) & 0.11820450 & 0.09406420 & 5000 & 1.65850145 \\
CME & 0.11992104 & 0.09378235 &  5000 & 0.74242733 \\
Cholesky  & 0.12130840 &  0.09308411 &  5000 & -\\
\hline
$d = 0.25$ & $\widehat{d}$ & $s$ & $N$ & $|t|$ statistic \\ \hline
wavelet ($g_0$ length  $2\times 1200+1$, $J=6$) & 0.27379919 & 0.09391947 & 5000 & 1.19695166\\
wavelet (CME at  $j = 0$,  $J=6$) & 0.27623054 & 0.09390228 & 5000 & 0.10736144\\
CME & 0.27474721 & 0.09585494 &  5000 & 0.68122236 \\
Cholesky  & 0.27603043 & 0.09248460 &  5000 & -\\
\hline
$d =0.45$ & $\widehat{d}$ & $s$ & $N$ & $|t|$ statistic \\ \hline
wavelet ($2\times 1200+1$, $J=6$) & 0.48484741 & 0.09462295 & 5000 & 1.04322102\\
wavelet ($2\times 1400+1$, $J=6$) & 0.48596691 & 0.09366874 & 5000 & 1.64788134\\
wavelet (CME at  $j = 0$,  $J=6$) & 0.48238424 & 0.09478310 & 5000 & 0.26842086 \\
CME & 0.48504209 & 0.09394773 &  5000 & 1.15107347\\
Cholesky  & 0.48288866 & 0.09313095  &  5000 & -\\
\hline
$d = -0.10$ & $\widehat{d}$ & $s$ & $N$ & $|t|$ statistic \\ \hline
wavelet (CME at  $j = 0$,  $J=6$) & $-0.06958397$ & 0.09376265 & 5000 & 1.87458648 \\
CME & $-0.07254477$ & 0.09363777  &  5000 & 0.28587914\\
Cholesky  & $-0.07307713$ & 0.09257674 &  5000 & -\\
\hline
$d = -0.25$ & $\widehat{d}$ & $s$ & $N$ & $|t|$ statistic \\ \hline
wavelet (CME at  $j = 0$,  $J=6$) & $-0.21858996$ & 0.09392950 & 5000 & 0.93216739\\
CME & $-0.21875428$ & 0.09171082 &  5000 & 1.03200374 \\
Cholesky  & $-0.21684272$ & 0.09350826  &  5000 & -\\
\hline
$d = -0.45$ & $\widehat{d}$ & $s$ & $N$ & $|t|$ statistic \\ \hline
wavelet (CME at  $j = 0$,  $J=6$) & $-0.40363919$ & 0.09449238   & 5000 & 1.07109047 \\
CME & $-0.40688947$ & 0.09378554 &  5000 &  0.64835641\\
Cholesky & $-0.40566672$ & 0.09480350 &  5000 & -\\

\end{tabular}\label{t:sim_fracOU_d_length9_initj0procCME}
\end{table}

\begin{table}[h]
\caption{Yule-Walker estimation of the AR(1) parameter $\phi = \exp(-1 \cdot  2^{-3}) = \exp(-2 \cdot 2^{-4}) = \exp(- 1/2 \cdot 2^{-2})= 0.8825$
(wavelet filters cut off at lag $|T| = 40$ or at value $10^{-9}$, time series length $2^{11}$ for all methods) }
\centering
\begin{tabular}{ccccc}
\hline
$\phi = 0.8825$ & $\widehat{\phi}$ & $s$ & $N$ & $|t|$ statistic \\ \hline
wavelet ($\zeta =1$, $J = 6$) & 0.88016217 & 0.01047668 & 5000 & 0.00300676\\
wavelet ($\zeta=2$, $J = 8$) &  0.88017601 & 0.01072520  & 5000 & 0.27145002 \\
wavelet ($\zeta=1/2$, $J = 4$) & 0.87919411 & 0.01072708 & 5000 & 4.34765925 \\
wavelet ($\zeta=1/2$, $J = 6$) & 0.87978815 & 0.01064935 & 5000 & 1.55891496\\
wavelet ($\zeta=1/2$, $J = 8$) & 0.88023961 & 0.01058475  & 5000 & 0.57449953\\
CME  & 0.88051697 & 0.01055337 & 5000 &  0.60412089\\
iterative  & 0.88011831 & 0.01052941 & 5000 & -
\end{tabular}\label{t:sim_OU_08825}
\end{table}

\newpage

\begin{table}[h]
\caption{fOU: Local Whittle estimation of $d$ ($\zeta=1$), comparison across values of $J$ (wavelet filters cut off at lag $|T| = 40$, time series length $2^{9}$ for all methods)}

\centering
\begin{tabular}{ccccc}
\hline
$d=0.25$ & $\widehat{d}$ & $s$ & $N$ & $|t|$ statistic \\ \hline
wavelet (CME at  $j = 0$,  $J=2$) & 0.27497237 & 0.09367323 & 5000 & 0.56835559\\
wavelet (CME at  $j = 0$,  $J=4$) & 0.27473605 & 0.09096327 & 5000 & 0.70556038\\
wavelet (CME at  $j = 0$,  $J=6$) & 0.27430310 & 0.09215599 & 5000 & 0.93550788\\
wavelet (CME at  $j = 0$,  $J=8$) & 0.27408226 & 0.09386958 & 5000 & 1.04538369\\
wavelet (CME at  $j = 0$,  $J=10$) & 0.27750365 & 0.09482862 & 5000 & 0.78643924 \\
Cholesky & 0.27603043 & 0.09248460 &  5000 & -\\
\hline
$d=-0.25$ & $\widehat{d}$ & $s$ & $N$ & $|t|$ statistic \\ \hline
wavelet (CME at  $j = 0$,  $J=2$) & $-0.21832655$ & 0.09358734 & 5000 & 0.79308646 \\
wavelet (CME at  $j = 0$,  $J=4$) & $-0.21692926$ & 0.09455304  & 5000 & 0.04601620\\
wavelet (CME at  $j = 0$,  $J=6$) & $-0.21911564$ & 0.09132082  & 5000 & 1.22965544\\
wavelet (CME at  $j = 0$,  $J=8$) & $-0.21530451$ & 0.09444787 & 5000 & 0.81837755\\
wavelet (CME at  $j = 0$,  $J=10$) & $-0.21680078$ & 0.09295220 & 5000 & 0.02249260 \\
Cholesky  & $-0.21684272$ & 0.09350826  &  5000 & -\\
\end{tabular}\label{t:sim_fracOU_length9_initj0procCME_compareJ}
\end{table}

\begin{table}[h]
\caption{fGLE: Local Whittle estimation of $d$ ($\zeta=2$, $m=1$), comparison across values of $J$ (wavelet filters for $\Delta X(n)$ cut off at lag $|T| = 80$; time series of length $2^{9}$ for all methods)
}
\centering\begin{tabular}{ccccc}
\hline
$d = 0.10$ & $\widehat{d}$ & $s$ & $N$ & $|t|$ statistic \\ \hline
wavelet (CME at  $j = 0$,  $J=6$) & $0.12044083$ &  0.09245007 & 5000 & 1.05431590\\
wavelet (CME at  $j = 0$,  $J=8$) & $0.12241433$ & 0.09102200 & 5000 & 0.02227624\\
wavelet (CME at  $j = 0$,  $J=10$) & $0.11932462$ & 0.09357999 & 5000 & 1.65282860\\
CME                         & $0.12221790$ & 0.09220211  &  5000  & 0.08515547\\
Cholesky  & $0.12237381$ & 0.09088329  &  5000 & -\\
\hline
$d = 0.25$ & $\widehat{d}$ & $s$ & $N$ & $|t|$ statistic \\ \hline
wavelet (CME at  $j = 0$,  $J=6$) & $0.27857299$ & 0.09217255  & 5000 & 0.25582876\\
wavelet (CME at  $j = 0$,  $J=8$) & $0.27998323$ & 0.09268102  & 5000 & 0.50777309\\
wavelet (CME at  $j = 0$,  $J=10$) & $0.28029380$ & 0.09627634 & 5000 & 0.66273074 \\
CME                         & $0.27818405$ &  0.09112314 &  5000  & 0.46947509\\
Cholesky  & $0.27904459$ &  0.09144735 &  5000 & -\\
\hline
$d = 0.45$ & $\widehat{d}$ & $s$ & $N$ & $|t|$ statistic \\ \hline
wavelet (CME at  $j = 0$,  $J=6$) & $0.43156246$ & 0.10538368  & 5000 & 0.78943448\\
wavelet (CME at  $j = 0$,  $J=8$) & $0.43199573$ &  0.10660877 & 5000 & 0.58088823\\
wavelet (CME at  $j = 0$,  $J=10$) & $0.43166818$ & 0.10678505 & 5000 & 0.73448769\\
CME                         & $0.43338219$ & 0.10403330  &  5000  & 0.07274950\\
Cholesky  & $0.43322954$ &  0.10579033 &  5000 & -\\
\end{tabular}\label{t:sim_fracGLE_length9_initj0procCME_compareJ}
\end{table}


\section{A study of the accuracy of numerical integration}\label{s:accuracy_numerical_integration}

We studied the accuracy of \texttt{quadl.m} by comparing it to the closed-form
$$
\psi_j = \frac{1}{2 \pi}\int^{\pi}_{-\pi}e^{ijx}(1 - \phi e^{-ix})^{-1}dx = \phi^{j} 1_{\{j \geq 0\}}, \quad - 1 < \phi < 1.
$$
For $\phi = -0.95$, $-0.50$, 0.50, 0.90, 0.95, and $|T| = 35$ or 100 depending on the decay of the filter, we obtained mean absolute deviation values $|\psi^{\textnormal{quadl}}_{t}- \psi_{t}|$ of the order of $10^{-17}$, which is comparable to machine precision.



In Table \ref{t:quadl_accuracy_FARIMA}, we also study the accuracy of \texttt{quadl.m} by comparing it to the closed-form
$$
\psi_j = \frac{1}{2 \pi}\int^{\pi}_{-\pi}e^{ijx}(1 - e^{-ix})^{-d}dx = \frac{\Gamma(j+d)}{\Gamma(d) \Gamma(j+1)}1_{\{j \geq 0\}}, \quad d \in \Big( -\frac{1}{2},\frac{1}{2}\Big)\backslash\{0\}.
$$
The value $\rho$ gives the radius of the ball around the singularity.
\begin{table}[h]
\caption{Numerical accuracy, FARIMA filters: \texttt{quadl.m} versus closed-form cut off at lag $|T|$, $\textnormal{deviation}_{t}$ = $\psi^{\textnormal{quadl}}_{t}- \psi_{t}$)}
\centering
\begin{tabular}{ccccc}
\hline
$d$ & $\rho$ & $|T|$ & mean abs.\ dev.\ ($|t| \leq |T|$) \\ \hline
$-0.45$ & $10^{-12}$ & 100 & 2.26 $\times$ $10^{-14}$\\
$-0.25$ & $10^{-12}$ & 100 & 3.01 $\times$ $10^{-14}$\\
$-0.10$ & $10^{-12}$ & 100 & 4.24 $\times$ $10^{-14}$\\
 0.10 & $10^{-12}$ & 100 & 5.56 $\times$ $10^{-12}$\\
 0.25 & $10^{-12}$ & 100 & 3.92 $\times$ $10^{-10}$\\
 0.45 & $10^{-11}$ & 100 &  3.92 $\times$ $10^{-7}$\\
\end{tabular}\label{t:quadl_accuracy_FARIMA}
\end{table}
The results show a high degree of accuracy in all depicted instances of FARIMA(0,$d$,0). As expected, the integration procedure becomes less accurate in the case of FARIMA as $d$ approaches 0.5. For $d = 0.45$, for instance, we picked $\rho = 10^{-11}$ due to convergence problems. To compute the covariance function of FARIMA(0,$d$,0)
$$
\gamma(j) = \frac{1}{2 \pi}\int^{\pi}_{-\pi}e^{ijx}|1 - e^{-ix}|^{-2d}dx = \frac{\Gamma(1-2d)}{\Gamma(d)\Gamma(1-d)}\frac{\Gamma(j+d)}{\Gamma(j-d+1)}1_{\{j \geq 0\}},
$$
$d \in \Big( -\frac{1}{2},\frac{1}{2}\Big)\backslash\{0\}$, \texttt{quadl.m} required a wider radius around the singularity at zero as $d$ approaches $0.5$ (not shown). The command \texttt{quadgk.m} is suitable when the singularity at zero behaves like $x^{p}$ for $p \geq -0.5$. The results are displayed in Table \ref{t:quadl_accuracy_FARIMAcov}. The deviation for the lag $t = 0$ is shown separately because it is the only one of a different order of magnitude.

\begin{table}[h]
\caption{Numerical accuracy, FARIMA autocovariance: \texttt{quadgk.m} versus closed-form (cut off at lags 1 and $T$, $\textnormal{deviation}_{t}$ = $\gamma^{\textnormal{quadl}}_{t}- \gamma_{t}$)}
\centering
\begin{tabular}{ccccc}
\hline
$d$ & $ T$ & mean abs.\ dev.\ ($1 \leq t \leq T$) & abs.\ dev.\ ($t = 0$) \\ \hline
 0.10 & 100 & 9.66 $\times$ $10^{-12}$ & 1.95 $\times$ $10^{-2}$\\
 0.15 & 100 & 1.98 $\times$ $10^{-11}$ & 4.88 $\times$ $10^{-2}$\\
 0.20 & 100 & 2.02 $\times$ $10^{-10}$ & 9.87 $\times$ $10^{-2}$\\
 0.25 & 100 & 1.36 $\times$ $10^{-13}$ & 1.80 $\times$ $10^{-1}$\\
\end{tabular}\label{t:quadl_accuracy_FARIMAcov}
\end{table}

\section{Pseudocode}\label{s:pseudocode}

{\small
\begin{center}
\begin{tabular}{|l|}
\hline
\\
Wavelet-based simulation over the interval $[0,1]$\\
\\
\hline
\\
\textbf{Input}: the parameter values of the stochastic process, the final scale $J$, \\
the number of zero moments of the underlying wavelet basis; \\
\\
\textbf{Step 1}: numerically calculate the Fourier transform of $\widehat{u}_{j}(x)$, $\widehat{v}_{j}(x)$, $j = 0,1,\hdots,J$, to obtain \\
the time domain filters $u_{j}$, $v_{j}$. For simplicity, let $L$ be the constant length of $u_j$, $v_j$;\\
\\
\textbf{Step 2}: numerically calculate the autocovariance function $r_0$ of the discrete time process \\ induced
by the filter $\widehat{g}_{0}(x)$ at $j = 0$;\\
\\
\textbf{Step 3}: generate an initialization sequence $V_0$ of size $L+1$ via CME with the autocovariance $r_0$;\\
\\
\textbf{Step 4}: generate a sequence of standard Gaussian white noise $\varepsilon_j$ of the same length \\ $(2^j-1) + (L+1)$ of $V_j$.
Insert zeroes between the entries of both vectors $V_j$, $\varepsilon_j$. Convolve the \\ resulting vectors
with the low- and high-pass filters $u_j$, $v_j$ to obtain \\ $u_j \ast \uparrow_2 V_j$, $v_j \ast\uparrow_2 \varepsilon_j$, respectively.\\
\\
\textbf{Step 5}: drop the $2(L-1)$ entries of the vectors $u_j \ast \uparrow_2 V_j$, $v_j \ast\uparrow_2 \varepsilon_j$ which are affected\\ by
the boundary effect (filter truncation)
and add together the two resulting vectors \\to obtain the vector $V_{j+1}$
 of length $(2^{j+1}-1) + (L+1)$.\\
\\
\textbf{Step 6}: stop if $j = J$. The resulting simulation vector $V_{J}$ has size $(2^{J} -1) +(L+1) > 2^J$.\\
\\
\textbf{Step 7}: else set $j \leftarrow j +1$ and goto \textbf{Step 4}.\\
\\
\hline
\end{tabular}
\end{center}
}

\begin{remark}
Simulation over the time interval $[0,2^{K}]$, $K \in \bbN$, can be performed in an analogous fashion.
\end{remark}


\bibliographystyle{agsm}

\bibliography{DidierFricks_wavelet_simulation_anomalous_diffusion_revised_arxiv}

\small

%
%

\end{document}